\documentclass[10pt]{article}
\usepackage{array}
\usepackage{faktor}
\usepackage{wrapfig, amsmath}
\usepackage{amsfonts}
\usepackage{amssymb}
\usepackage{bm}
\usepackage{graphicx}
\usepackage{epsfig,cancel,amsthm, amssymb, todonotes, amsmath}
\usepackage{color,tikz}
\usetikzlibrary{decorations.markings,arrows, calc}
\usepackage{color}
\usepackage{graphicx,framed,verbatim,caption,amsbsy}
\usepackage[colorlinks=true, pdfstartview=FitV, linkcolor=darkblue, citecolor=darkblue, urlcolor=darkblue]{hyperref}
\usepackage{enumitem}
\usepackage[rflt]{floatflt}
\usepackage{float}

\definecolor{shadecolor}{rgb}{0.9, 0.9, 0.86}
\definecolor{darkgreen}{rgb}{0.2, 0.5,  0}
\definecolor{darkblue}{rgb}{0.1,0.1,0.45}

\def\&{\vspace{-5pt}&}

\def\Im{\mathrm {Im}\,}

\def\e{{\rm e}}
\def\dd{{\rm d}}
\newcommand\norm[1]{\left\lVert#1\right\rVert}

\usepackage{tikz}
\usepackage{pgf}
\usepackage{pgflibraryarrows}
\usepackage{pgflibrarysnakes}
\usetikzlibrary{decorations.text}
\usepgfmodule{shapes}
\usetikzlibrary{decorations.pathmorphing}
\usetikzlibrary{decorations.markings}
\usetikzlibrary{patterns}
\usetikzlibrary{automata}
\usetikzlibrary{positioning}

\tikzset{->-/.style={decoration={
 markings,
 mark=at position #1 with {\arrow{>}}},postaction={decorate}}}

\def \eqref#1{(\ref{#1})}
\def \& {&\hspace{-10pt}}

\renewcommand{\d}{\mathrm d}
       
\newtheorem{theorem}{Theorem}[section]
\newtheorem{example}[theorem]{Example}
\newtheorem{exercise}[theorem]{Exercise}

\newtheorem{lemma}[theorem]{Lemma}
\newtheorem{remark}[theorem]{Remark}
\newtheorem{problem}[theorem]{Riemann--Hilbert Problem}

\newtheorem{proposition}[theorem]{Proposition} 
\newtheorem{corollary}[theorem]{Corollary} 
 
\newtheorem{definition}[theorem]{Definition}
\newtheorem{assumption}[theorem]{Assumption}

\def\d{{\rm d}}

\def\tr{{\mathrm{Tr}}}

\def\bt{\begin{theorem}}
\def\et{\end{theorem}}
\def\bc{\begin{corollary}}
\def\ec{\end{corollary}}
\def\bx{\begin{example}}
\def\ex{\end{example}}
\def\bxr{\begin{exercise}\small}
\def\exr{\end{exercise}}
\def\bl{\begin{lemma}}
\def\el{\end{lemma}}
\def\bd{\begin{definition}}
\def\ed{\end{definition}}
\def\bp{\begin{proposition}}
\def\ep{\end{proposition}}

\def\br{\begin{remark}}
\def\er{\end{remark}}

\def\be{\begin{equation}}
\def\ee{\end{equation}}

\def\&{\hspace{-15pt}&}
\def\bea{\begin{eqnarray}}
\def\eea{\end{eqnarray}}
\def\beas{\begin{eqnarray*}}
\def\eeas{\end{eqnarray*}}

\def\1{{\bf 1}}

\def\i{\mathrm i}
\DeclareMathOperator{\im}{Im}

\makeatletter
\@addtoreset{equation}{section}
\makeatother

\newcommand{\cmnew}[1]{\textcolor{violet}{#1}}

 \usepackage{authblk}

\date{}                     

\title{Unitary ensembles with a critical edge point, their multiplicative statistics and the Korteweg-de-Vries hierarchy}

\author[1]{Mattia Cafasso}
\author[2]{Carla Mariana da Silva Pinheiro}
\affil[1]{\textit{Univ Angers, CNRS, LAREMA, SFR MATHSTIC, F-49000 Angers, France ;} \texttt{mattia.cafasso@univ-angers.fr}}
\affil[2]{\textit{Instituto de Ciências Matemáticas e de Computação -  Universidade de São Paulo;} \texttt{carla.pinheiro@usp.br}}
\begin{document}
\maketitle

\begin{abstract}
We study the multiplicative statistics associated to the limiting determinantal point process describing eigenvalues of unitary random matrices with a critical edge point, where the limiting eigenvalue density vanishes like a power 5/2. We prove that these statistics are governed by the first three equations of the KdV hierarchy, and study the asymptotic behavior of the relevant solutions.
\end{abstract}

\tableofcontents \hfill\\


\newpage



	

\section{Introduction and statement of the results}

Determinantal point processes were introduced by O. Macchi in 1975 \cite{Macchi}. They appear naturally in the study of fermionic systems, random matrices, random permutations, tilings and many other different models, see \cite{Soshnikov} and \cite{JohanssonDPP} for a pedagogical introduction. In this paper, we are interested in a particular point process, describing the universal behavior of eigenvalues of unitary random matrix ensembles near a critical edge point, where the limiting eigenvalue density vanishes like power 5/2. In the physics literature, this kind of multicritical models begun to be studied in the nineties \cite{BoBre, BreMaPa}. The relevant kernel was introduced, in the mathematical literature, by Claeys and Vanlessen \cite{CV07}. It is of integrable type (in the sense of Its-Izergin-Korepin and Slavnov \cite{IIKS}), as it is written in the form
\begin{equation}\label{eq:introkernel}
K(\lambda,\mu;t_0,t_1) := \frac{\phi_1(\lambda;t_0,t_1)\phi_2(\mu;t_0,t_1) - \phi_2(\lambda;t_0,t_1)\phi_1(\mu;t_0,t_1)}{-2 \pi \i(\lambda -\mu)}.
\end{equation}
The functions $\phi_j, j=1,2$ are best described as entries of the (unique) solution of a given Riemann-Hilbert problem, depending parametrically on $t_0, t_1 \in \mathbb R$, and are related to a distinguished solution $y = y(t_0,t_1)$ of the second member of the Painlev\'e I hierarchy, usually denoted by $PI^{(2)}$. This is a fourth-order analogue of the Painlevé I equation, and it reads
\begin{equation}
	\frac{\partial_{t_0}^4y}{64} +\frac{5}{8}\left((\partial_{t_{0}}y)^2 + 2y\partial_{t_0}^2y\right)+10y^3+4y t_1 - 4 t_{0}=0.
\end{equation}
The Painlev\'e trascendent $y$ plays a distinguished role in the so-called Dubrovin's conjecture about the universality of the critical behavior of Hamiltonian perturbation of hyperbolic PDEs, see \cite{DubConjecture} and also \cite{GravaClaeys2012} for a partial solution of the conjecture restricted to the equations of the KdV hierarchy, which are defined below. The precise definition of $y$, as well as the one of $\phi_j, j = 1,2$, is formulated in terms of the Riemann-Hilbert problem \ref{RHP1} below, in Section \ref{sec:2}.
The kernel \eqref{eq:introkernel} should be thought of as a higher-order analogue of the Airy kernel, the latter appearing near regular edge points of  unitary random matrices, when the limiting eigenvalue density vanishes like power 1/2 (as, for example, in the case of the Wigner semi-circle law, associated to the Gaussian Unitary Ensemble). As for the Airy point process, one can prove that the point process $\mathcal X$ associated to \eqref{eq:introkernel} has (almost surely) a largest particle and, as such, almost all realisations of the point process are given by a strictly decreasing sequence of points $\xi_0 > \xi_1 > \xi_2 > \cdots$. We introduce a real function $\sigma : \mathbb R \to \mathbb R$, satisfying the following assumptions: 
 \begin{assumption}\label{assumption1}\hfill
 	\begin{enumerate}
 		\item $\sigma : \mathbb R \to [0,1]$ is a non-decreasing function, $C^\infty$ everywhere except, possibly, at a finite number of points $r_1, \ldots, r_k$. At those points, the left and right limits of $\sigma$ and of all its derivatives exist (but are not, in general, the same).
		\item The function $r \to r^4\sigma(r)$ is in $L^2(\mathbb R^-, \d r)$.  
	\end{enumerate}
 \end{assumption}
We will denote $\iota := \lim_{r \to \infty} \sigma(r)$ (because of the assumption above, such a $\iota$ exists and $0 \leq \iota \leq 1$).

We are interested in:
\begin{equation}\label{def:Qsigma}
	Q_{\sigma}(t_0,t_1,s_0,s_1) := \mathbb E \left[ \prod_{j \geq 0} \left( 1 - \sigma\left(\frac{\xi_j}{s_0^{2/7}} - \frac{s_1}{s_0}\right)\right)\right],
\end{equation}
where $s_0 > 0$ and $s_1 \in \mathbb R$. These two parameters $s_0$ and $s_1$ parametrise affine transformations of the point process $\mathcal X$. 
\begin{remark}
One might wonder about the peculiar choice of the parametrisation of affine transformations in \eqref{def:Qsigma}. Besides being convenient in the sequel of the paper, they are also related to the (integrated) density of the point process $\mathcal X$. Using the corollary \ref{cor:density} in Appendix \ref{AppendixPI}, heuristically we expect that
$$
	\mathbb E \left[ \# \left\{ j \in \mathbb N \; : \; \xi_j > -r \right\} \right] \sim \int_{-r}^0 \frac{1}{\pi} \rho^{\frac{5}2} \mathrm{d} \rho =  \frac{2}{7\pi} r^{7/2},
$$
so that the expected number of points, to the right of $-r$, for the rescaled point process is inversely proportional to $s_0$:
$$
	\mathbb E \left[ \# \left\{ j \in \mathbb N \; : \; \frac{\xi_j}{s_0^{2/7}} > -r \right\} \right]  \sim \frac{2}{7\pi s_0} r^{7/2}.
$$
A similar parametrisation, related to the density of the Airy point process, was also used in \cite{CCR}, see equation (1.4) in loc. cit. 
\end{remark}

\begin{remark}
The quantities $Q_{\sigma}(t_0,t_1,s_0,s_1)$ are usually called \emph{multiplicative statistics} associated to the point process. In the last years, due to their applications to solvable models in the KPZ universality class and to integrable PDEs, they attracted considerable attention by researchers interested in integrability. They have the following elegant interpretation in terms of \emph{thinning} of a point process, see for instance \cite{ClaeysGlesner} and references therein. Consider a realisation $\xi_0 > \xi_1 > \xi_2 > \cdots$ of the point process $\mathcal X$ and construct a new configuration of points by keeping each $\xi_j, \, j \geq 0,$ with probability $\sigma(\xi_j s_0^{-2/7} - s_1s_0^{-1})$, and by deleting it with probability $1 - \sigma(\xi_j s_0^{-2/7} - s_1s_0^{-1})$. In this way, a new point process $\mathcal X^{\sigma}$ is created by thinning the original one. $Q_{\sigma}(t_0,t_1,s_0,s_1) $ is the probability that $\mathcal X^{\sigma}$ is empty:
$$
Q_\sigma(t_0,t_1,s_0,s_1) := \mathbb P\left( \mathcal X^{\sigma} = \emptyset\right).
$$
\end{remark}
Since $\mathcal X$ is a \emph{determinantal} point process, by general results about determinantal point processes (see for instance \cite{JohanssonDPP}), $Q_{\sigma}(t_0,t_1,s_0,s_1)$ are Fredholm determinants, see equations and \eqref{eq:QasFredDet} and \eqref{eqintro:FredSeries} below. Moreover, being the associated kernel of \emph{integrable type}, in the sense of \cite{IIKS}, a natural connection is established with Riemann-Hilbert problems and, ultimately, with integrable PDEs. The ones relevant in this work are the first three equations of the KdV hierarchy, which we now define.
\begin{definition}
Let $R_k$ be the Lenard differential polynomials, depending on $v$ and its derivatives with respect to $\tau_1$, defined by the recursion relation
\begin{equation}\label{LenardRecursion}
	D_1 \mathcal{R}_{2k+3}[v] = \left( \frac{1}4 D_1^3 + 2vD_1 + v_{\tau_1} \right) \mathcal{R}_{2k + 1}[v], \quad \quad \mathcal{R}_1[v] := v,
\end{equation}
 where $D_k := \frac{\partial}{\partial \tau_k}$ and $v_{\tau_k} = D_k v, \quad k \geq 1$.
 The KdV hierarchy is, by definition, the set of commuting flows given by the equations
 \begin{equation}\label{KdVhierarchy}
 	v_{\tau_{2k + 1}} = D_1 \big(\mathcal{R}_{2k + 1}[v] \Big), \quad k \geq 0.
 \end{equation}
 \end{definition}
It is easy to see that, when $k = 1$, the equation \eqref{KdVhierarchy} becomes the celebrated KdV equation \eqref{eq:kdv1}. This equation, as firstly observed in 1968 in a seminal paper by Gardner, Kruskal and Miura \cite{GKM68},  possesses an infinite number of conserved quantities, each one giving an equation of the hierarchy, defined by the flow with respect to the parameter $\tau_{2k + 1}$. The recursion \eqref{LenardRecursion}, which is already present in \cite{GGKM}, equation (3.20), has been successively interpreted by Magri \cite{Magri} as a bihamiltonian recursion, $D_1$ and $\frac{1}4 D_1^3 + 2vD_1 + v_{\tau_1}$ being the two (compatible) Poisson operators associated to the equation.

In order to state our first main result, we need to introduce the change of variable $\{t_0,t_1,s_0,s_1 \} \mapsto \{\tau_1,\tau_3,\tau_5,\tau_7 \}$ defined by the one-to-one map
\begin{equation}\label{eq:changeofvar}
	\tau_1 = \displaystyle \frac{5 s_1^3}{8 s_0^2} + \frac{t_1 s_1}{s_0^{4/7}} - 2t_0s_0^{1/7}, \quad  \tau_3 = \displaystyle \frac{5 s_1^2}{4 s_0} + \frac{2 t_1 s_0^{3/7}}3, \quad \tau_5 = s_1, \quad \tau_7 = \displaystyle \frac{2}7 s_0.
\end{equation}
Our first main result states that $Q_{\sigma}(\tau_1,\tau_3,\tau_5,\tau_7)$ produces a family of solutions, depending on the function $\sigma$, to the first three equations of the Korteweg-de-Vries hierarchy, and also to three more equations detailed in the theorem below.
 
\begin{theorem}\label{thm:Main1}
	Let $\sigma : \mathbb R \longrightarrow [0,1]$ be a function satisfying Assumptions \ref{assumption1} and $$y = y(t_0(\tau_1,\tau_3,\tau_5,\tau_7), t_1(\tau_1,\tau_3,\tau_5,\tau_7)), \quad \quad  h = h(t_0(\tau_1,\tau_3,\tau_5,\tau_7),t_1(\tau_1,\tau_3,\tau_5,\tau_7))$$ 
	defined as in \eqref{eq:asympPhi}. Then 
\begin{enumerate}
	\item The function
	\begin{equation}\label{def:v}
		v = v(\tau_1,\tau_3,\tau_5,\tau_7) := \frac{\partial^2}{\partial \tau_1^2} \log Q_\sigma + \left(\frac{2}{7 \tau_7}\right)^{2/7}\!\!\!\!\!\! y - \frac{1}7 \frac{\tau_5}{\tau_7}
	\end{equation}
	satisfies the first three equations of the KdV hierarchy:
	\begin{equation}\label{eq:KdVhierarchy}
		v_{\tau_{2k + 1}} = D_1 \Big(\mathcal{R}_{2k + 1}[v] \Big), \quad k = 1,2,3.
	\end{equation}
	\item The function 
	\begin{equation}\label{def:u}
		u = u(\tau_1,\tau_3, \tau_5, \tau_7) := \frac{\partial}{\partial \tau_1} \log Q_\sigma - \left(\frac{2}7\right)^{1/7}\!\!\!\!\frac{h}{\tau_7^{1/7}} + \frac{3}{98}\frac{\tau_3 \tau_5^2}{\tau_7^2} - \frac{15}{2744}\frac{\tau_5^4}{\tau_7^3} - \frac{1}7\frac{\tau_1\tau_5}{\tau_7},
	\end{equation}
	satisfies the ``potential KdV'' equations
	\begin{equation}\label{eq:pKdVhierarchy}
		u_{\tau_{2k + 1}} =  \mathcal{R}_{2k + 1}[u_{\tau_1}], \quad k = 1,2,3.
	\end{equation}
\end{enumerate}
\end{theorem}
\begin{remark}
The functions $v$ and $u$, equations \eqref{def:v} and \eqref{def:u}, are well defined because, under Assumption \ref{assumption1}, $Q_\sigma(t_0,t_1,s_0,s_1) > 0$, see Remark \ref{rem:nonzerodets} below.
\end{remark}	
	
In a more explicit form, denoting with a prime the derivative with respect to $\tau_1$, the equations satisfied by $v$, defined in \eqref{def:v}, are
	\begin{align}
		v_{\tau_3} =&\frac{1}{4}v'''+3 v v', \label{eq:kdv1}\\
		v_{\tau_5} =& \frac{1}{16}v^{(\mathrm v)}+\frac{5}{2}v'v''+\frac{5}{4}vv'''+\frac{{15}}{2}v^2v', \label{eq:kdv2}\\
		v_{\tau_7} = & \frac{1}{64} v^{(\mathrm{vii})} + \frac{35}{16}v''v''' + \frac{35}2 vv'v'' + \frac{7}{16}vv^{(\mathrm{v})} + \frac{35}2 v^3v' + \frac{35}8 v'^3 + \frac{21}{16}v'v^{(\mathrm{iv})} + \frac{35}8 v^2 v'''. \label{eq:kdv3}
	\end{align}
while those of $u$, defined in equation \eqref{def:u}, are
\begin{align}
		u_{\tau_3} =&\frac{1}{4}u'''+ \frac{3}2 (u')^2, \label{eq:pkdv1}\\
		u_{\tau_5} =& \frac{1}{16}u^{(\mathrm v)}+\frac{5}{4}u'u'''+\frac{5}{8}(u'')^2+\frac{5}2 (u')^3, \label{eq:pkdv2}\\
		u_{\tau_7} = & \frac{1}{64} u^{(\mathrm{vii})} + \frac{7}{16}u'u^{(\mathrm{v})} + \frac{21}{32} (u''')^2 + \frac{7}8 u''u^{(\mathrm{iv})} + \frac{35}8 u'''(u')^2 + \frac{35}8 u' (u'')^2 + \frac{35}8 (u')^4.  \label{eq:pkdv3}
	\end{align}
	The first part of Theorem \ref{thm:Main1} above should be regarded as the analog, for the Claeys-Vanlessen kernel, of Theorem 1.3 in \cite{CCR} for the Airy kernel. The asymptotic results of \cite{CCR} have been further improved in \cite{ChCR}. As other recent examples of multiplicative statistics associated to integrable evolution equations, let us also mention \cite{CT} for the sine kernel, \cite{RuzzaBessel} for the Bessel kernel, and \cite{CR,CMR} for the discrete Bessel kernel. Multiplicative statistics of determinantal point processes are also associated to integro-differential versions of Painlev\'e equations, and can be described using operator-valued Lax systems; see for instance \cite{BotRew,Kraj,BCT, BL}.

\begin{remark}\label{rem:sigmaHeaviside}
	Consider the particular case in which $\sigma = H$, with $H$ the Heaviside step function. In this case, the function $Q_H$ in \eqref{def:Qsigma} reduces to the cumulative distribution function associated to the largest particle in the point process:
	\begin{equation}\label{eq:QH}
		Q_H(\tau_1,\tau_3,\tau_5,\tau_7) = \mathbb P\left(\xi_1 \leq \frac{\tau_5}{\left(\frac{7}2 \tau_7\right)^{\frac{5}{7}}}\right).
	\end{equation}
In other words, $Q_H$ depends just on a particular combination of $\tau_5$ and $\tau_7$ and so does $v$, defined in \eqref{def:v}. By consequence, it is easy to prove that, for each fixed $\tau_5$ and $\tau_7$, we have that
$$\tau_5 v_{\tau_5} + \frac{7}5 \tau_7 v_{\tau_7} = 0.$$
Using the definition of the KdV hierarchy and integrating once, we find that
\begin{equation}\label{eq:statred}
	5\tau_5 \mathcal R_5[v] + 7\tau_7 \mathcal R_7[v] = \mathrm{const},
\end{equation}
which is known as a \emph{stationary KdV reduction}. These equations, in the periodic case, have been studied since the seventies, in relation with finite-gap solutions of the KdV equation, see \cite{Nov74,Dub75,ItsMatveev75}. In the non-periodic case, stationary reductions are related to multi-solitons solutions and their stability; see \cite{MaddSachs}.
On the other hand, the Claeys-Vanlessen kernel \eqref{eq:introkernel} is the second one (the first one being the Airy kernel) of a hierarchy of kernels, whose gap probability $Q_H$ has been studied in \cite{CIK}, and put in relation with the Painlev\'e II hierarchy. Even in the case of $\sigma = H$, the results found in this paper are different from the ones in \cite{CIK}. Indeed, in this paper $v$ is the double log-derivative of $Q_\sigma$ with respect to $\tau_1$, while Painlevé type equations are obtained in \cite{CIK} differentiating $Q_H$ by the combination of the variables $\tau_5$ and $\tau_7$ appearing in \eqref{eq:QH}; see nevertheless Remark \ref{rem:asymptotics} below for a comparaison between the asymptotics obtained in the present paper and in \cite{CIK} for $n = 2$. We will study the multiplicative statistics associated to the whole family of kernels in \cite{CIK} in a subsequent work.
\end{remark}

\begin{remark}\label{rem:KdV}
	It can be verified (see details in Section \eqref{sec:2}) that 
	\begin{equation}\label{eq:ytilde}
		\tilde y(\tau_1,\tau_3,\tau_5,\tau_7) := \left(\frac{2}{7 \tau_7}\right)^{2/7}\!\!\!\!\!\! y - \frac{1}7 \frac{\tau_5}{\tau_7}
	\end{equation}
	is itself a solution of \eqref{eq:kdv1}. Analogously, \begin{equation}\label{eq:htilde}
\tilde h(\tau_1,\tau_3,\tau_5,\tau_7) : = \left(\frac{2}7\right)^{1/7}\!\!\!\!\frac{h}{\tau_7^{1/7}} - \frac{3}{98}\frac{\tau_3 \tau_5^2}{\tau_7^2} + \frac{15}{2744}\frac{\tau_5^4}{\tau_7^3} + \frac{1}7\frac{\tau_1\tau_5}{\tau_7}
\end{equation}
is itself a solution of \eqref{eq:pkdv1}.
\end{remark}
	


The second main result concerns the (singular) behavior of the functions $v$ and $u$ when $\tau_7 \to 0$. From now on, we will work under a stronger assumption for $\sigma$:
\begin{assumption}\label{assumption2}
	There exist $k_1, k_2, k_3, K >0$ positive real constants such that
	\begin{equation}
		|\sigma(z) - \iota\chi_{(0,\infty)}(z)| \leq k_1 {\rm e}^{-k_2|z|^3}, \quad \forall z \in \mathbb R
	\end{equation}
	\begin{equation}
		|\sigma'(z))| \leq k_3|z|^{-2}, \quad \forall z \in \mathbb R, \quad |z| > K.
	\end{equation}
\end{assumption}

\begin{theorem}\label{thm:asympuvgamma}
	Consider $Q_\sigma(\tau_1,\tau_3,\tau_5,\tau_7)$ as defined in \eqref{def:Qsigma}, with $\sigma$ satisfying Assumptions \ref{assumption1} and \ref{assumption2}, and $u$ and $v$ defined respectively in \eqref{def:u} and \eqref{def:v}. Then:
	\begin{enumerate}
		\item Let $\tau_5 \geq M (\tau_7)^{5/7}$, for a fixed real constant $M > 0$ and $\tau_1,\tau_3$ satisfying the Assumptions \ref{assumptionsfirstregime}. Then, there exist two positive constants $T$ and $k$ such that
		\begin{align}
		u =& - \left(\frac{2}{7\tau_7}\right)^{1/7}h - \left(\frac{\tau_{1} \tau_{5}}{7\tau_{7}} - \frac{3 \tau_{3} \tau_{5}^{2}}{98\tau_{7}^{2}} + \frac{15 \tau_{5}^{4}}{2744 \tau_{7}^{3}}\right) +O (\e^{-k \tau_5/\tau_7^{5/7}}),\\
		v =&\, \left(\frac{2}{7\tau_7}\right)^{2/7}y-\frac{\tau_5}{7\tau_7}+O (\e^{-k \tau_5/\tau_7^{5/7}}), \label{vasymp1}
\end{align}
uniformly for $\tau_7 \in (0,T)$, where 
$$h = h(t_0(\tau_1,\tau_3,\tau_5,\tau_7),t_1(\tau_1,\tau_3,\tau_5,\tau_7)) \quad \text{and} \quad y = y(t_0(\tau_1,\tau_3,\tau_5,\tau_7),t_1(\tau_1,\tau_3,\tau_5,\tau_7)),$$ defined in \eqref{eq:asympPhi}, are respectively distinguished solutions to the potential KdV equation \eqref{eq:pKdV} and to the $PI^{(2)}$ equation \eqref{eq:PI2}.
		\item Let $|\tau_k| \leq M \tau_7^{k/7}$, $k = 1,3,5$ for a large enough constant $M$. Then, as $\tau_7 \to 0$, 
\begin{align}
u &= \left(\frac{2^6}{7\tau_7}\right)^{1/7}\left(\frac{q}{2}+ p \right)+ O(\tau_7^{1/7})\\
v &= \left(\frac{2^6}{7\tau_7}\right)^{1/7}\left(\frac{\partial}{\partial \tau_1}\frac{q}{2}- \left(\frac{2^6}{7\tau_7}\right)^{1/7}\frac{q^2}{2} + O\left(\tau_7^{1/7}\right)\right), \label{vasymp2}
\end{align}
where $q = q(x_0,x_1,x_2)$ is a solution to the third member of the Painlevé II hierarchy \eqref{PII3} and the first two equations of the modified KdV hierarchy \eqref{mKdV1} \eqref{mKdV2}, and  $\frac{\partial}{\partial x_0} p = \frac{q^2}2$. The variables $x_0,x_1,x_2$ are given, in function of $\tau_j, \; j = 1,3,5,$ by
\begin{equation}
	x_0 = - a \frac{\tau_1}{\tau_7^{1/7}}, \quad x_1 = \frac{3a^3}{4}  \frac{\tau_3}{\tau_7^{3/7}}, \quad x_2 = \frac{5a^5}{16} \frac{\tau_5}{\tau_7^{5/7}}.
\end{equation}
where $a=(2^6/7)^{1/7}$. Moreover,  $p$ and $q$ are both characterized by  the uniquely solvable Riemann-Hilbert Problem \ref{RHPP2} with $s = \i(\iota-1)$, see equation \eqref{eq:pandq}.

		\item Suppose that $\iota = 1$. Let $M_1,M_2$ and $M_3$ be positive real constants such that $-M_1 \tau_7^{4/7} \leq \tau_5 \leq -M_2 \tau_7^{5/7}$, $| \tau_3 | \leq M_1 \tau_7^{2/7}$, $-M_3 \leq \tau_1 \leq x$ where 
$$x=-\tau_7\left(\frac{2}{7\tau_7}\right)^{6/7}+\tau_5\left(\frac{2}{7\tau_7}\right)^{4/7}-|\tau_3|\left(\frac{2}{7\tau_7}\right)^{2/7},$$ 
and $(-M_3,x) \neq \emptyset$. Then, as $\tau_7 \to 0$,
\begin{align}
u &= \cmnew{-} p_{\sigma}(-\tau_1) + O(x),\\
v &= -p_{\sigma}^{2}(-\tau_1) \cmnew{+} 2 q_{\sigma}(-\tau_1) + O(x), \label{vsigma}
\end{align}
where $p_\sigma$ and $q_\sigma$, which depend solely on $\tau_1$, are characterized by the uniquely solvable Riemann-Hilbert problem \ref{RHPHsigma}.
\end{enumerate}
\end{theorem}
\begin{remark}
One should note the similarities between Theorem \ref{thm:asympuvgamma} and the analog theorem for the Airy point process, Theorem 1.8 in \cite{CCR}. In particular, the quantity $-p_{\sigma}^{2}(-\tau_1) + 2 q_{\sigma}(-\tau_1) $ in \eqref{vsigma} is the same as $v_\sigma$ in  Theorem 1.8 in \cite{CCR}, equation (1.25)\footnote{In fact, $v_\sigma$ is defined in \cite{CCR} by means of two functions $q_{\sigma}^0, p_{\sigma}^0$ as $v_{\sigma} = -(p_{\sigma}^0)^{2}-2q_{\sigma}^0$, and through equation \eqref{eq:connectCCR} one can easily verify that $q_{\sigma}^0(.)$ in \cite{CCR} is equal to $-q_{\sigma}(.)$ in this paper.}, see details in Section \ref{sec:6}. Moreover, in both the first and second asymptotic regimes, the asymptotic approximations of $v$, equations \eqref{vasymp1} and \eqref{vasymp2}, are themselves solutions of the KdV equation \eqref{eq:kdv1}, as it was already happening in \cite{CCR}. For the first asymptotic regime, this has been already observed in Remark \ref{rem:KdV}. For the second asymptotic regime, one recognizes that
$$\tilde{q}(\tau_1, \tau_3, \tau_5, \tau_7) := \left(\frac{2^6}{7\tau_7}\right)^{1/7}\left(\frac{\partial}{\partial \tau_1}\frac{q}{2}- \left(\frac{2^6}{7\tau_7}\right)^{1/7}\frac{q^2}{2} \right)$$
is, essentially, a Miura transformation of a solution of modified KdV, plus a re-scaling of coordinates. More specifically, from equation \eqref{mKdV1} we know that $q$ solves $\frac{\partial}{\partial \tau_3}q=-6 \left(\frac{1}{14 \tau_7} \right)^{2/7}q^2\frac{\partial}{\partial \tau_1}q+\frac{1}{4}\frac{\partial^3}{\partial \tau_1^3}q$, and this allows us to conclude that $\tilde{q}$ satisfies equation \eqref{eq:kdv1}.
\end{remark}

\begin{remark}\label{rem:asymptotics}
Let us consider again the special case in which $\sigma = H$ the Heaviside function, as in Remark \ref{rem:sigmaHeaviside}. We compare our asymptotic results with the ones obtained in \cite{CIK} (for $n = 2$) for the "higher Tracy-Widom distribution" describing the position of the largest particle in the point process $\mathcal X$. We start observing that, as $s_1/s_0^{5/7} \to \infty$, a close inspection of the proof of Lemma \ref{lemma:SNcase1} shows that
$$\frac{\partial}{\partial t_0} \log Q_{\sigma} = O\left(\e^{-c(s_1/s_0^{5/7})^{7/2}}\right),
$$
in agreement with \cite{CIK}.

On the other side, Theorem \ref{thm:Main1} implies that (in the $PI^{(2)}$ variables),
\begin{align*}
\frac{\partial}{\partial s_1}\frac{\partial}{\partial t_0} \log Q_{\sigma} =&- \frac{5 s_{1}^{3} }{8 s_{0}^{20/7}} - \frac{s_{1} t_{1}}{s_{0}^{10/7}} + \frac{2t_{0}}{s_{0}^{5/7}} - 2 s_{0}^{1/7} \frac{\partial}{\partial s_1} u. 
\end{align*}
Moreover, it can be checked that, in the third asymptotic regime of Theorem \ref{thm:asympuvgamma} with $x \leq -\delta$ for some fixed $\delta>0$,
$$ \frac{\partial}{\partial s_1} u = \mathrm{const}+O(s_0^{1/7})$$ 
where $\mathrm{const} \approx 1.1341$, as it can be verified from the asymptotic expansion of the Bessel Riemann-Hilbert problem stated in equation \eqref{eq:defPhiBe}. Consequently, as $s_0 \to 0$,
\begin{align}
\frac{\partial}{\partial s_1}\frac{\partial}{\partial t_0} \log Q_{\sigma} =&- \frac{5 s_{1}^{3} }{8 s_{0}^{20/7}} - \frac{s_{1} t_{1}}{s_{0}^{10/7}} + \frac{2t_{0}}{s_{0}^{5/7}} + O(s_{0}^{1/7}).\label{eq:dels1selt0Q}
\end{align}
This result is consistent with \cite{CIK}. In fact,  Theorem 1.5 from \cite{CIK} claims that, as $s_1/s_0^{5/7} \to -\infty$,
\begin{align*}
\frac{\partial}{\partial s_1} \log Q_{\sigma} =&\frac{25 s_{1}^{6}}{256 s_{0}^{5}} + \frac{5 s_{1}^{4} t_{1}}{16 s_{0}^{25/7}} - \frac{5 s_{1}^{3} t_{0}}{8 s_{0}^{20/7}} + \frac{s_{1}^{2} t_{1}^{2}}{4 s_{0}^{15/7}} - \frac{s_{1} t_{0} t_{1}}{s_{0}^{10/7}} + \frac{t_{0}^{2}}{s_{0}^{5/7}} - \frac{3}{8 s_{1}}+ O\left(s_0^{10/7}/s_1^3\right).
\end{align*}
Formally differentiating with respect to $t_0$, one obtains the same leading behaviour as in Equation \eqref{eq:dels1selt0Q}.

\end{remark}

\section{The Riemann-Hilbert characterization of $Q_\sigma$}\label{sec:2}
We start recalling the Riemann-Hilbert problem related to the kernel \eqref{eq:introkernel}. We slightly reformulate the one described in \cite{CIK}, page 3 and 4 and $k = 1$ (see also the previous work \cite{CV07}, Section 3.6.1, but with different normalization), ``flattening'' the set of contours on the real axis.
\begin{problem}\label{RHP1}\hfill
\begin{enumerate}
\item $ \Phi(t_0,t_1;z) \equiv \Phi(z)$ is analytic on $\mathbb{C}\setminus \mathbb R$, and has continuous boundary values $\Phi_\pm$ satisfying the jump relation
\begin{equation}
\Phi_+(z)=
\Phi_-(z)
\begin{pmatrix}
1 & 1 \\ 0 & 1
\end{pmatrix},  z\in \mathbb R.
\end{equation}
\item As $z\to \infty$,
\begin{equation}\label{eq:asympPhi}
\begin{split}\Phi(z)=z^{-\sigma_3/4}N\left(I+\dfrac{1}{z^{1/2}}\begin{pmatrix}
-h & 0 \\ 0 & h
\end{pmatrix}+\dfrac{1}{2z}\begin{pmatrix}
h^2 & \i y \\ -\i y & h^2
\end{pmatrix}+ \sum_{j\geq 2}\dfrac{\Phi^{(j)}}{z^{j/2}} \right)
\e^{-\theta(z)\sigma_3} \\ \times \left\lbrace \begin{array}{cc}
I \hspace{0.3cm}&  |\arg z|<\pi-\epsilon\\
\begin{pmatrix}
1 & 0 \\ \pm 1 & 1
\end{pmatrix}  \hspace{0.3cm}& \pi-\epsilon<\pm \arg z<\pi  
\end{array}\right.,\end{split}
\end{equation}
where $N := \frac{1}{\sqrt{2}}\begin{pmatrix} 1 & 1 \\ -1 & 1	\end{pmatrix}\e^{-\pi \i \sigma_3/4}$ and $\theta(z) := \frac{2}{7} z^{7/2} + \frac{2}{3} t_1 z^{3/2} - 2 t_0 z^{1/2}$. In the formula above, $0 < \epsilon < \pi/2$, and the principal branches of $z^{-\sigma_3/4}$ and $z^{1/2}$ are taken, analytic in $\mathbb C/(-\infty, 0]$ and positive for $z > 0.$ 
\end{enumerate}
\end{problem}

\begin{remark}
The scalars $h, y$ and the matrices $\Phi^{(j)}$, even if not explicitly written, are functions of the parameters $t_0$ and $t_1$.
\end{remark}

\begin{remark}
It was proven in \cite{CV06} that the Riemann-Hilbert problem above is uniquely solvable, when $t_0$ and $t_1$ are real. Using symmetries of the Riemann-Hilbert problem, it is also proven, in the same paper, that the matrix-valued functions $\Phi^{(j)}$ in the asymptotic expansion of the Riemann-Hilbert problem \ref{RHP1} present the following structure,
\begin{align}\label{eq:symmetriesPhi}
\Phi^{(2k+1)} = \begin{pmatrix}
q_k& \i r_k\\
\i r_k & -q_k
\end{pmatrix}, \qquad \Phi^{(2k+2)} = \begin{pmatrix}
v_{k}& \i w_{k}\\
-\i w_{k} & v_{k}
\end{pmatrix},
\end{align}
where the functions $q_k, r_k,v_k$ and $w_k$ are real-valued.
\end{remark}
Differentiating the solution $\Phi$ of the Riemann-Hilbert problem above with respect to $t_0,t_1$ and $z$, one can write a Lax system which, in turn, gives differential equations for $y = y(t_0,t_1)$. In particular, we have the following

\begin{proposition}\label{prop:eqPI2} 
 The functions $h = h(t_0,t_1)$ and $y = y(t_0,t_1)$ defined in equation \eqref{eq:asympPhi} satisfy the relation 
\begin{equation}\label{eq:handy}
  \frac{\partial}{\partial t_0} h = 2 y.
 \end{equation}
 The function $h$ satisfy the potential KdV equation
\begin{equation}\label{eq:pKdV}
	\partial_{t_1}h + \frac{1}4 \left( \partial_{t_0} h\right)^2 + \frac{1}{48} \partial_{t_0}^3 h = 0,
\end{equation}
 while $y$ solves the $PI^{(2)}$ and the KdV equation
\begin{equation}\label{eq:PI2}
 	\frac{1}{64} \partial_{t_0}^4y +\frac{5}{8}\left((\partial_{t_{0}}y)^2 + 2y\partial_{t_0}^2y\right)+10y^3+4y t_1 - 4 t_{0}=0,
\end{equation}
\begin{equation}\label{eq:KdV}
	\partial_{t_1} y + y\partial_{t_0} y + \frac{1}{48}\partial_{t_0}^3 y = 0.
\end{equation} 
\end{proposition}

\begin{remark}
We draw the attention of the reader to the fact that both the KdV equation and the potential KdV equation appear twice in the paper, with two different normalizations: see respectively equations \eqref{eq:pkdv1}, \eqref{eq:pKdV} and \eqref{eq:kdv1} and \eqref{eq:KdV}. For completeness, we report here the necessary rescaling bringing one equation into the other:
\begin{align*}
v(\tau_1, \tau_3) = \kappa_0 y(\kappa_1 \tau_1, \kappa_3\tau_3),
\end{align*}
where $\kappa_0=2^{2/5}3^{2/5}, \; \kappa_1 = 2^{4/5}3^{1/5}$, and $\kappa_3 = -\kappa_0$.
$$
u(\tau_1, \tau_3) = \nu_0h(\nu_1 \tau_1, \nu_3\tau_3),
$$
where $\nu_0 = -(3/2)^{1/4}$, $\nu_1=-2^{3/4}3^{1/4}$ and $\nu_3 =-\nu_0$. 
\end{remark}

The Proposition above is well known (see, for instance, \cite{CV07}, \cite{CIK} and references therein), except perhaps for the equation \eqref{eq:pKdV}. For the readers' convenience, we re-derived them in Appendix \ref{AppendixPI}. Then, it is straightforward to check, using the change of coordinates \eqref{eq:changeofvar}, combined with equations \eqref{eq:pKdV} and \eqref{eq:KdV},  that $\tilde y$ and $\tilde h$, defined in Remark \ref{rem:KdV}, satisfy the Korteweg-de-Vries and potential Korteweg-de-Vries equations \eqref{eq:kdv1} and \eqref{eq:pkdv1}.

\begin{definition}
	Let $\Phi$ be the unique solution of the Riemann-Hilbert problem \ref{RHP1}. We define
	\begin{equation}\label{eq:defphi}
		\phi_1(\lambda;t_0,t_1) := \Big(\Phi_+(\lambda;,t_0,t_1)\Big)_{11}, \quad \quad \phi_2(\lambda;t_0,t_1) := \Big(\Phi_+(\lambda;t_0,t_1)\Big)_{21},
	\end{equation}
	and the Claeys-Vanlessen kernel 
	\begin{equation}\label{secchapter:introkernel}
		K(\lambda,\mu;t_0,t_1) := \frac{\phi_1(\lambda;t_0,t_1)\phi_2(\mu;t_0,t_1) - \phi_2(\lambda;t_0,t_1)\phi_1(\mu;t_0,t_1)}{-2 \pi \i(\lambda -\mu)}.
	\end{equation}
\end{definition}
\begin{remark}\label{rem:compactkernel}
Using the jumps of the Riemann-Hilbert problem \ref{RHP1}, one can verify that the kernel $K(\lambda,\mu;t_0,t_1)$ is actually independent from which boundary value of $\Phi$ we are taking in equation \eqref{eq:defphi}. This is easily seen, for instance, rewriting the kernel $K(\lambda,\mu;t_0,t_1) \equiv K(\lambda,\mu)$ as
$$
	K(\lambda,\mu) = \frac{1}{2 \pi \i} \frac{\begin{pmatrix} 1 & 0 \end{pmatrix} \Phi_+^T(\lambda)\Phi_+^{-T}(\mu) \begin{pmatrix} 0 \\ 1\end{pmatrix}}{\lambda - \mu}.
$$
\end{remark}

Standard arguments in the theory of determinantal point processes (see for instance \cite{JohanssonDPP}) allow us to express the multiplicative statistics $Q_\sigma$, defined in \eqref{def:Qsigma}, as a Fredholm determinant. To the point, define $f(\lambda) = f(\lambda;t_0,t_1,s_0,s_1)$ and $h(\lambda) = h(\lambda;t_0,t_1,s_0,s_1)$ as
\begin{equation}
f(\lambda) := \frac{1}{\sqrt{\pi i}}\sqrt{\sigma\left(\frac{\lambda}{s_0^{2/7}}-\frac{s_1}{s_0}\right)}\begin{pmatrix}
\phi_1(\lambda) \\ \phi_2(\lambda)
\end{pmatrix}, \quad h(\lambda) := \dfrac{1}{2\sqrt{\pi i}}\sqrt{\sigma\left(\frac{\lambda}{s_0^{2/7}}-\frac{s_1}{s_0}\right)}\begin{pmatrix}
-\phi_2(\lambda) \\ \phi_1(\lambda).
\end{pmatrix},
\label{eq:rhp:1}
\end{equation}
and the associated integrable kernel as
\begin{equation}
K_{\sigma}(\lambda,\mu;t_0,t_1,s_0,s_1) := \frac{f^T(\lambda)h(\mu)}{\lambda-\mu}.
\end{equation}

Then,

\begin{equation}\label{eq:QasFredDet}
	Q_\sigma(t_0,t_1,s_0,s_1) = \det(I - \mathbb K_{\sigma}),
\end{equation}
where $\mathbb K_\sigma$ is the integral operator associated to the kernel $K_\sigma$, whose Fredholm determinant is equal to
\begin{equation}\label{eqintro:FredSeries}
	\det(I - \mathbb K_{\sigma}) = 1 + \sum_{k \geq 1} (-1)^k \int_{\mathbb R^k} \det(K_\sigma(\lambda_i,\lambda_j)) \d \lambda_1 \cdots \d \lambda_k.
\end{equation}

\begin{remark}\label{rem:nonzerodets}
The decaying condition we imposed in Assumption \ref{assumption1} for $\sigma$ (see point 2. of the Assumption) and the fact that $K(\lambda,\lambda) = O(|\lambda|^{5/2})$ as $\lambda \to -\infty$ (see Appendix \ref{AppendixPI}) implies that $K_\sigma(\lambda,\lambda) \in L^1(\mathbb R)$ so that $\mathbb K_\sigma$ is a (non-negative) trace-class operator, and hence the Fredholm series \eqref{eqintro:FredSeries} is convergent.\\
We can also prove that $Q_\sigma(t_0,t_1,s_0,s_1) > 0$ for any $t_0,t_1,s_1 \in \mathbb R$ and $s_0 > 0$. To see this, we prove that $1$ is not an eigenvalue of $\mathbb K_\sigma$.  Indeed, suppose that $f \in L^2(\mathbb R)$ is such that $\mathbb K_\sigma f = f$. We now show that $f = 0$. We start rewriting $\mathbb K_\sigma := \mathbb M \mathbb K \mathbb M$, where $\mathbb K$ is the projection operator (for example, because of equation (1.35) in \cite{CV07})  associated to the Claeys-Vanlessen kernel \eqref{eq:introkernel}, and $\mathbb M$ is the multiplication operator by $\sigma^{1/2}\left(\lambda s_0^{-2/7}- s_1 s_0^{-1}\right)$ . 
Now set $g := \mathbb K \mathbb M f$. Using the eigenvalue equation $\mathbb K_\sigma f = f$, we see that $g$ is an entire function such that $f = \mathbb M g$. Hence,
\begin{align*}
\norm{g}_{L^2(\mathbb{R})} = \norm{\mathbb K \mathbb M f}_{L^2(\mathbb{R})} \leq \norm{\mathbb M f}_{L^2(\mathbb{R})} \leq \norm{f}_{L^2(\mathbb{R})} = \norm{\mathbb M g}_{L^2(\mathbb{R})} \leq \norm{g}_{L^2(\mathbb{R})}.
\end{align*}

Consequently, we have that actually $$\norm{\mathbb M g}_{L^2(\mathbb{R})} = \norm{g}_{L^2(\mathbb{R})},$$
and this implies that $g \equiv 0$. This is easily seen when $\iota = \lim_{r \to +\infty} < 1$. In the case $\iota = 1$, because of the decaying condition of $\sigma$ at $-\infty$, we still have that there exists $\lambda_0$ large enough such that $g(\lambda) = 0$ for all $\lambda \leq \lambda_0$, and this implies $g = 0$ because of analyticity. Since $f = \mathbb M g$, we proved that $f \equiv 0$.
\end{remark}

Note, in particular, that $K_{\sigma}$ is of integrable type \cite{IIKS}, and as such it is related to a Riemann-Hilbert problem in the way we now describe (for a pedagogical presentation of the results used below, see \cite{BDS}). First of all, recall that, thanks to Remark \ref{rem:nonzerodets}, $Q_\sigma(t_0,t_1,s_0,s_1) \in (0,1]$. Hence, one can define the resolvent operator
\begin{equation}\label{eq:defres}
	\mathbb L_\sigma := (I - \mathbb K_\sigma)^{-1}\mathbb K_\sigma,
\end{equation}
which is also of integrable type. Now, let $Y(z)$ be the $2 \times 2$ matrix-valued function 
\begin{equation}\label{eq:defY}
	Y(z) = I - \int_{\mathbb{R}} \dfrac{F(\lambda)h^T(\lambda)}{\lambda-z} \dd \lambda,
\end{equation}
where 
\begin{equation}\label{eq:F}
F(\lambda) = (1 - \mathbb{K}_{\sigma})^{-1} f(\lambda),
\end{equation}
($(1 - \mathbb{K}_{\sigma})^{-1} $ acts entry-wise on $f$). Then, $Y$ solves the Riemann-Hilbert problem \ref{RHP2} below and, additionally, the kernel of the resolvent $L_\sigma$
is given by
\begin{equation}
	L_\sigma(\lambda,\mu, t_0,t_1,s_0,s_1) = \frac{1}{\lambda - \mu}f^T(\lambda) Y_+^T(\lambda)Y_+^{-T}(\mu)h(\mu),
\end{equation}
where $Y_+$ is the boundary value of $Y$.

\begin{problem}\label{RHP2}\hfill
\begin{enumerate}
\item $Y = Y(z;t_0,t_1,s_0,s_1)$ is analytic on $\mathbb{C}\setminus \mathbb{R}$.
\item $Y$ has boundary values $Y_{\pm}$ on $\mathbb R$ such that $Y_\pm - I \in L^2(\mathbb R)$, and which are continuous except possibly at the points $(\frac{r_j}{s_0^{2/7}}-\frac{s_1}{s_0}), \; j = 1,\cdots,k$, for $r_j$ the same as in Assumption \ref{assumption1}, and related by the jump relation
$$
Y_+(z)=
Y_-(z)[I-2\pi \i f(z)h^T(z)], \quad \forall z \in \mathbb R.
$$
\item As $z\to \infty$,
$$
Y(z)=I+\dfrac{1}{z}Y^{(1)}+\dfrac{1}{z^2}Y^{(2)}+O(z^{-3}).
$$
\end{enumerate}
\end{problem}

A straightforward calculation allows us to rewrite the jump in the previous Riemann-Hilbert problem as follows
\begin{equation}\label{eq:JY}
Y_+(z)=
Y_-(z)\begin{pmatrix}
1+\sigma\left(z s_0^{-2/7}-s_1/s_0\right)\phi_1(z)\phi_2(z) & -\sigma\left(z s_0^{-2/7}-s_1/s_0\right)\phi_1(z)^2 \vspace{0.3cm}\\ \sigma\left(z s_0^{-2/7}-s_1/s_0\right)\phi_2(z)^2 & 1-\sigma\left(z s_0^{-2/7}-s_1/s_0\right)\phi_1(z)\phi_2(z)
\end{pmatrix}.
\end{equation}
Moreover, it is easy to see that $\det(Y(z)) \equiv 1$, and hence the matrix $Y^{(1)}$ in the asymptotic expansion of the Riemann-Hilbert problem \ref{RHP2} can be written as
\begin{equation}\label{eq:Y1}
	Y^{(1)}(t_0,t_1,s_0,s_1) = \begin{pmatrix}
							\alpha(t_0,t_1,s_0,s_1) & \gamma(t_0,t_1,s_0,s_1) \\
							\beta(t_0,t_1,s_0,s_1) & -\alpha(t_0,t_1,s_0,s_1) 
						\end{pmatrix}.
\end{equation}
A very explicit formula relates $Y^{(1)}$ to the first log-derivative of $Q_\sigma$ with respect to $\tau_1$.

\begin{proposition}
	For any $\tau_1,\tau_3,\tau_5 \in \mathbb R$ and $\tau_7 \in \mathbb R_{>0}$ 
	\begin{equation}\label{eq:logderivativeQ}
		\frac{\partial}{\partial \tau_1} \log Q_\sigma(\tau_1,\tau_3,\tau_5,\tau_7) = -\left( \frac{2}{7 \tau_7}\right)^{1/7}\gamma(\tau_1,\tau_3,\tau_5,\tau_7),
	\end{equation}
	where $\gamma(\tau_1,\tau_3,\tau_5,\tau_7)$ is defined in \eqref{eq:Y1}.
\end{proposition}
\begin{proof}
	The Jacobi's formula implies that
\begin{equation}
\partial_{t_0} \log Q_\sigma = \partial_{t_0} \log \det(1 - \mathbb{K}_{\sigma}) = - \tr \left((1 - \mathbb{K}_{\sigma})^{-1}\partial_{t_0} \mathbb{K}_{\sigma}\right).
\end{equation}

Recall that the kernel $K_\sigma$ of $\mathbb K_\sigma$ is written as
$$K_{\sigma}(\lambda,\mu) = \frac{f^T(\lambda)h(\mu)}{\lambda-\mu},
$$
where $f$ and $h$ are associated to the solution $\Phi$ of the Riemann-Hilbert problem \ref{RHP1}, see equation \eqref{eq:rhp:1}. 

In particular, using the relation 
$$\partial_{t_0}\Phi =-2 \begin{pmatrix}
0&1\\ z-2y&0
\end{pmatrix} \Phi,
$$
(see Appendix \ref{AppendixPI})
we compute 
\begin{align*}
\partial_{t_0} K_{\sigma}(\lambda,  \mu) =& -\frac{1}{\pi i} \phi_1(\lambda)\phi_1(\mu)\sqrt{\sigma\left(\frac{\lambda}{s_0^{2/7}}-\frac{s_1}{s_0}\right)}\sqrt{\sigma\left(\frac{\mu}{s_0^{2/7}}-\frac{s_1}{s_0}\right)}\\
=&-[f(\lambda)]_{11} [f(\mu)]_{11}.
\end{align*}
(in the last equality, we used equation \eqref{eq:F}).
Therefore,
\begin{equation}
- \tr \left((1 - \mathbb{K}_{\sigma})^{-1}\partial_{t_0} \mathbb{K}_{\sigma}\right) = \langle (1 - \mathbb{K}_{\sigma})^{-1} [f]_{11}, [f]_{11}\rangle = \langle [F]_{11}, [f]_{11}\rangle
\end{equation}
where the brackets $<\cdot, \cdot>$ denote the standard inner product on $L^2(\mathbb R)$.

We now use equation \eqref{eq:defY} that, combined with the 
asymptotic expansion for $Y(z)$ as $z \to \infty$, allows us to conclude that
\begin{equation}
Y^{(1)}= \int_{\mathbb{R}}F(s)h^T(s) \dd s.
\label{eq:expy1}
\end{equation}

A direct consequence of equation \eqref{eq:expy1}, using \eqref{eq:rhp:1}, is that
$$Y^{(1)}_{12} = \frac{1}{2}\langle [F]_{11}, [f]_{11}\rangle,
$$ 
and, therefore,
$$\partial_{t_0} \log Q_\sigma(\tau_1,\tau_3,\tau_5,\tau_7)  = 2 Y^{(1)}_{12}(\tau_1,\tau_3,\tau_5,\tau_7) = 2 \gamma(\tau_1,\tau_3,\tau_5,\tau_7).
$$
The result now follows by considering $\frac{\partial}{\partial \tau_1} = -\frac{1}{2}\left(\frac{2}{7\tau_7}\right)^{1/7}\!\!\!\frac{\partial}{\partial t_0}$.
\end{proof}

We now combine the two Riemann-Hilbert problems \ref{RHP1} and \ref{RHP2} to simplify the jumps of the latter, which at the end will allow us to find the related integrable equations and study the asymptotic behavior of $Q_\sigma(t_0,t_1,s_0,s_1)$. 
To this end, define the new spectral variable 
\begin{equation}\label{eq:defzeta}
	\zeta := z s_0^{-2/7}-s_1/s_0.
\end{equation}
We will work with the new variables $\tau_j,\; j = 1,3,5,7$, related to $(t_0,t_1,s_0,s_1)$ by the equations \eqref{eq:changeofvar}.
We define the sectionally-analytic matrix-valued function
\begin{equation}\label{eq:defX}
	X(\zeta) \equiv X(\zeta;\tau_1,\tau_3,\tau_5,\tau_7) := A s_0^{\frac{\sigma_3}{14}}Y(z(\zeta))\Phi(z(\zeta)),
\end{equation} 
where
\begin{equation}\label{def:A}
	A := \begin{pmatrix}
	1 & 0 \\ \gamma s_0^{-1/7}& 1
\end{pmatrix}.
\end{equation}

We now show that $X(\zeta)$ solves the following Riemann-Hilbert problem:
\begin{problem}\label{RHP3}\hfill
\begin{enumerate}
\item $X$ is analytic on $\mathbb{C}\setminus \mathbb{R}$.
\item $X$ has boundary values $X_{\pm}$ on $\mathbb R$ such that $X_{\pm} \in L^2_{\mathrm{loc}}(\mathbb R)$, which are continuous except (possibly) at the points $r_1,\ldots,r_k$ and satisfy the jump relation
$$
X_+(\zeta)=
X_-(\zeta)\begin{pmatrix}
1 & 1-\sigma(\zeta) \\ 0 & 1
\end{pmatrix}.
$$
\item As $\zeta\to \infty$,
\begin{align}\label{eq:asympX}
X(\zeta) = & A\zeta^{-\frac{\sigma_3}{4}}N\left[I+\sum_{j\geq 1}\dfrac{X^{(j)}}{\zeta^{j/2}}\right]\e^{\varphi(\zeta)\sigma_3}\left\lbrace \begin{array}{cc}
I \hspace{0.3cm}&  |\arg \zeta|<\pi-\epsilon\\
\begin{pmatrix}
1 & 0 \\ \pm 1 & 1
\end{pmatrix}  \hspace{0.3cm}& \pi-\epsilon<\pm \arg \zeta<\pi  
\end{array}\right.,
\end{align}
where $\varphi(\zeta):= -\left(\tau_7 \zeta^{\frac{7}{2}}+\tau_5\zeta^{\frac{5}{2}}+\tau_3\zeta^{\frac{3}{2}}+\tau_1\zeta^{\frac{1}{2}}\right)$. In the formula above, $0 < \epsilon < \pi/2$, and the principal branches of $\zeta^{-\sigma_3/4}$ and $\zeta^{1/2}$ are taken, analytic in $\mathbb C\backslash (-\infty, 0]$ and positive for $\zeta > 0.$
\end{enumerate}
\end{problem}
The conditions (1) and (2) follow from the analogous conditions for $Y$ and $\Phi$; as for the explicit expression of the jump condition, it comes from the following chain of inequalities:
\begin{align*}
	X^{-1}_-(\zeta) X_+(\zeta) = \Phi_-^{-1}(z(\zeta))Y_-^{-1}(z(\zeta))Y_+(z(\zeta))\Phi_+(z(\zeta)) = \\ 
	= \Phi_-(z(\zeta))\Big( I - \sigma(\zeta)\Phi_+(z(\zeta)) \begin{pmatrix} 1 \\ 0\end{pmatrix} \begin{pmatrix} 1 & 0 \end{pmatrix} \Phi_+^{-1}(z(\zeta)) \Big)\Phi_+(z(\zeta)) = \begin{pmatrix} 1  & 1 - \sigma(\zeta) \\ 0 & 1 \end{pmatrix},
\end{align*}
where, in the last equality, we used that 
$
\Phi_+(z(\zeta)) \begin{pmatrix} 1 \\ 0\end{pmatrix} = \Phi_-(z(\zeta)) \begin{pmatrix} 1 \\ 0\end{pmatrix} .
$
The following proposition will be useful in the derivation of the Lax matrices associated to the solution of the Riemann-Hilbert problem \ref{RHP3}, described in the next section.

\begin{proposition}\label{prop:propertiesX}
	The matrices $X^{(j)}$ in the asymptotic expansion \eqref{eq:asympX} satisfy the symmetry conditions
	$$
		X^{(2k + 1)} = \begin{pmatrix} X_{11}^{(2k+1)} &  X_{11}^{(2k+1)} \\ 
		X_{12}^{(2k+1)} & -X_{11}^{(2k+1)}\end{pmatrix}, \quad
		X^{(2k)} = \begin{pmatrix} X_{11}^{(2k)} &  X_{12}^{(2k)} \\ 
		-X_{12}^{(2k)} & X_{11}^{(2k)}\end{pmatrix}.
	$$
Moreover, 
\begin{align*}
X_{11}^{(1)} &= \frac{\gamma}{2s_{0}^{1/7}} + c_1 + \frac{ h}{s_{0}^{1/7}} & X_{11}^{(2)}=& \frac{\gamma c_1}{2 s_{0}^{1/7}} + \frac{ \gamma h}{2 s_{0}^{2/7}} + \frac{c_1^{2}}{2} + \frac{ c_1 h}{s_{0}^{1/7}} + \frac{h^{2}}{2 s_{0}^{2/7}} \\
X_{12}^{(1)} &= \frac{i \gamma}{2s_{0}^{1/7}} &
X_{12}^{(2)}=&\i \left(\frac{ \gamma c_1}{2 s_{0}^{1/7}} + \frac{ \gamma h}{2s_{0}^{2/7}} + \frac{ \alpha}{s_{0}^{2/7}} - \frac{ s_{1}}{4s_{0}} + \frac{ y}{2s_{0}^{2/7}}\right)
\end{align*}
where
$$c_1 := \frac{s_{1}^{2} t_{1}}{4 s_{0}^{11/7}} - \frac{s_{1} t_{0}}{s_{0}^{6/7}} + \frac{5 s_{1}^{4}}{64 s_{0}^{3}}$$
and $\alpha$ and $\gamma$ are defined in equation \eqref{eq:Y1}, while $h$ and $y$ are defined in \eqref{eq:asympPhi}. 
\end{proposition}
\begin{proof}
	The symmetry conditions on the matrices $X^{(j)}$ come from the analog symmetry condition for $\Phi$, see equation \eqref{eq:symmetriesPhi}. The explicit expressions of the first two terms is derived from the definition \eqref{eq:defX} of $X$, together with the two asymptotic expansions of the Riemann-Hilbert problems \ref{RHP1} and \ref{RHP2} at $z = \infty$.  
\end{proof}

\section{KdV, potential KdV equations and Lax matrices}
The evolution equations \eqref{eq:kdv1}--\eqref{eq:kdv3} and \eqref{eq:pkdv1}--\eqref{eq:pkdv3},  satisfied by the variables $v$ and $u$ (equations \eqref{def:v} and \eqref{def:u}), are deduced by a set of Lax equations satisfied by the solution $X$ to the Riemann-Hilbert problem \ref{RHP3}. The starting point is the following Lemma which, formally, is easily verified using the properties of $X$: namely its asymptotic expansion at infinity and the fact that its jumps do not depend on the parameters $\tau_j$, $j = 1,3,5,7$, see the Riemann-Hilbert problem \ref{RHP3}. The rigorous analytical proof, modulo the obvious changes of notation, can be carried out exactly as in \cite{CCR}, Lemma 3.3.

\begin{lemma}
	For any $z \in \mathbb C\backslash \mathbb R$, the solution $X$ to the Riemann-Hilbert problem \ref{RHP3} is differentiable with respect to the variables $\tau_j, \; j = 1,3,5,7$. Moreover, its derivatives satisfy the Riemann-Hilbert problem detailed below, where $\partial$ denotes the derivative with respect to one of the variables $\tau_j$.
\end{lemma}

\begin{problem}\label{RHP4}\hfill
\begin{enumerate}
\item $\partial X$ is analytic on $\mathbb{C}\setminus \mathbb{R}$.
\item $\partial X$ has boundary values $\partial X_{\pm}$ on $\mathbb R$ such that $\partial X_{\pm} \in L^2_{\mathrm{loc}}(\mathbb R)$, which are continuous except (possibly) at the points $r_1,\ldots,r_k$ and satisfy the jump relation
$$
\partial X_+(\zeta)=
\partial X_-(\zeta)\begin{pmatrix}
1 & 1-\sigma(\zeta) \\ 0 & 1
\end{pmatrix}.
$$
\item As $\zeta\to \infty$,
\begin{align}\label{eq:dasympX}
\partial X(\zeta) = & A\zeta^{-\frac{\sigma_3}{4}}N\left[\left(I+\sum_{j\geq 1}\dfrac{X^{(j)}}{\zeta^{j/2}}\right)\partial \varphi(\zeta)\sigma_3 + \left(\sum_{j\geq 1}\dfrac{\partial X^{(j)}}{\zeta^{j/2}}\right)\right]\e^{\varphi(\zeta)\sigma_3}\\ \nonumber
&\hspace{5cm}
\times \left\lbrace \begin{array}{cc}
I \hspace{0.3cm}&  |\arg \zeta|<\pi-\epsilon\\
\begin{pmatrix}
1 & 0 \\ \pm 1 & 1
\end{pmatrix}  \hspace{0.3cm}& \pi-\epsilon<\pm \arg \zeta<\pi  
\end{array}\right..
\end{align}
In the formula above, $0 < \epsilon < \pi/2$, and the principal branches of $\zeta^{-\sigma_3/4}$ and $\zeta^{1/2}$ are taken, analytic in $\mathbb C\backslash(-\infty, 0]$ and positive for $\zeta > 0.$
\end{enumerate}
\end{problem}

\begin{proposition}
	The matrix-valued function $X(\zeta)$ defined in \eqref{eq:defX} satisfies the Lax equations
$$\frac{\partial}{\partial \tau_{2k + 1}}X = B_{2 k + 1} X = \begin{pmatrix}
-A_k & C_k \\ D_k & A_k
\end{pmatrix}X,
$$
where the entries $A_k,C_k,D_k$ of the Lax matrices are differential polynomials of $v$, defined in \eqref{def:v}. Explicitly, 
\begin{align}
C_k =& \zeta^k + \sum_{j=1}^k \zeta^{k-j} \mathcal{R}_{2j-1}[v] \label{eq:Ck}\\
A_k =& \frac{1}{2} \frac{\partial}{\partial \tau_1} C_k\label{eq:Ak}\\
D_k =& \zeta^{k+1}-\mathcal{R}_1[v]\zeta^k+ \sum_{j=2}^k \left\{ \mathcal{R}_{2j-1}[v]- \left(\frac{1}{2} \frac{\partial^2}{\partial \tau_1^2}+2 \mathcal{R}_1[v]\right) \mathcal{R}_{2j-3}\right\}\zeta^{k+1-j} \label{eq:Dk}\\
&- \left(\frac{1}{2} \frac{\partial^2}{\partial \tau_1^2}+2 \mathcal{R}_1[v]\right)\mathcal{R}_{2k-1}[v] \label{eq:Dk}.\nonumber
\end{align}
with $\mathcal{R}_{2j-1}$ defined as in \eqref{LenardRecursion}. 
Moreover, $v$ satisfies the KdV equations \eqref{eq:KdVhierarchy} while $u$, defined in \eqref{def:u} satisfies the potential KdV equations \eqref{eq:pKdVhierarchy}.
\end{proposition}
\begin{proof}
Consider the matrix valued functions
$$
	B_{2k + 1} := \left(\frac{\partial X}{\partial \tau_{2k + 1}} \right) X^{-1}, \quad k = 0,1,2,3.
$$
Using the Riemann-Hilbert characterization of $X$ and its derivatives, we deduce that these functions have no jumps on $\mathbb R/\{r_1, \ldots, r_k\}$. Hence, they are analytic functions with, possibly, isolated singularities at the points of discontinuity of $\sigma$. On the other hand, the boundary values of $X$ and its derivatives are locally $L^2$, and hence $B_{2k + 1}$ are, in fact, entire functions. More precisely, they are polynomials (in the variable $\zeta$) of degree $k + 1$, as easily deduced from the asymptotic expansion at infinity
\begin{align*}
\left(\frac{\partial X}{\partial \tau_{2k + 1}}\right)X^{-1} = & \begin{pmatrix}
0 & 0 \\ -2i\frac{\partial X_{12}^{(1)}}{\partial \tau_{2k + 1}}&0
\end{pmatrix}+ A\zeta^{-\frac{\sigma_3}{4}}N\left[\zeta^{-1/2}\frac{\partial X^{(1)}}{\partial \tau_{2k + 1}}+\zeta^{-1}\frac{\partial X^{(2)}}{\partial \tau_{2k + 1}} +\zeta^{-3/2}\frac{\partial X^{(3)}}{\partial \tau_{2k + 1}}+O\left(\zeta^{-2}\right)\right]\times \\
&\left[I+\dfrac{X^{(1)}}{\zeta^{1/2}}+\dfrac{X^{(2)}}{\zeta}+\dfrac{X^{(3)}}{\zeta^{3/2}}+O\left(\zeta^{-2}\right)\right]^{-1}N^{-1}\zeta^{\frac{\sigma_3}{4}}A^{-1}+\\
& A\zeta^{-\frac{\sigma_3}{4}}N\left[I+\dfrac{X^{(1)}}{\zeta^{1/2}}+\dfrac{X^{(2)}}{\zeta}+\dfrac{X^{(3)}}{\zeta^{3/2}}+O\left(\zeta^{-2}\right)\right](-\zeta^{k + 1/2}\sigma_3)\times \\
&\left[I+\dfrac{X^{(1)}}{\zeta^{1/2}}+\dfrac{X^{(2)}}{\zeta}+\dfrac{X^{(3)}}{\zeta^{3/2}}+O\left(\zeta^{-2}\right)\right]^{-1}N^{-1}\zeta^{\frac{\sigma_3}{4}}A^{-1}.
\end{align*}
We now observe that $u$, defined in \eqref{def:u}, can be equivalently rewritten as $u = X_{11}^{(1)} + \i X_{12}^{(1)}$, using equation \eqref{eq:logderivativeQ} and Proposition \ref{prop:propertiesX}. Hence, using the asymptotic expansion above, it can be verified that the functions $B_{2k + 1}$ satisfy the recursion relation 
 $$B_{2k + 1} = \zeta\left(B_{2k - 1} + \begin{pmatrix} 0 & 0 \\ u_{\tau_{2k - 1}} & 0 \end{pmatrix}\right) + R^{(2k + 1)},$$
 where  $R^{(2k + 1)}$ are matrices of degree $0$ in $\zeta$. Consequently, there exist functions $\{ b_j, j = 0, \ldots, 10 \}$, depending on the parameters $\{\tau_k, k = 1,3,5,7 \}$, such that
 \begin{align*}
B_1 &= \begin{pmatrix}
0 & 1 \\
\zeta+b_0&0
\end{pmatrix}\\
B_3 &= \begin{pmatrix}
b_1 & \zeta+u_{\tau_1} \\
\zeta^2-u_{\tau_1} \zeta+b_2&-b_1
\end{pmatrix}\\
B_5 &= \begin{pmatrix}
b_1\zeta + b_3 & \zeta^2+u_{\tau_1}\zeta +b_4\\
\zeta^3-u_{\tau_1} \zeta^2+b_5 \zeta +b_6&-b_1\zeta - b_3
\end{pmatrix} \\
B_7 &= \begin{pmatrix}
b_1\zeta^2+b_3\zeta+b_7 & \zeta^3+u_{\tau_1}\zeta^2 +b_4\zeta +b_8\\
\zeta^4-u_{\tau_1} \zeta^3+b_5 \zeta^2 +b_9\zeta +b_{10}&-b_1\zeta^2-b_3\zeta - b_7
\end{pmatrix},
\end{align*}
with $b_5 = b_2 + u_{\tau_3}$ and $b_9 = b_6 + u_{\tau_5}$. We now use 
 the Lax equations
\begin{equation}
\frac{\partial}{\partial \tau_1}B_k-\frac{\partial}{\partial \tau_k}B_1 = [B_1, B_k], \quad k = 3,5,7,
\label{eq:compcond}
\end{equation}
to obtain more information on these functions. For each $k$, we get three equations for the variables $\{ b_j, \; j = 0,\ldots,10 \}$ and one additional equation satisfied by $u$. Explicitly,
\begin{itemize}
	\item for $k =3$
		\begin{align}\label{b012}
			b_0 &= - 2\frac{\partial}{\partial \tau_{1}} u, &
			b_{1} &= - \frac{1}{2}\frac{\partial^{2}}{\partial \tau_{1}^{2}} u, &
			b_{2} &=- 2 \left(\frac{\partial}{\partial \tau_{1}} u\right)^{2} - \frac{1}{2}\frac{\partial^{3}}{\partial \tau_{1}^{3}} u
		\end{align}
	and 
		\begin{equation}\label{eq:proofkdv1}
		6 \frac{\partial}{\partial \tau_{1}} u \frac{\partial^{2}}{\partial \tau_{1}^{2}} u + \frac{1}{2}\frac{\partial^{4}}{\partial \tau_{1}^{4}} u - 2 \frac{\partial^{2}}{\partial \tau_{3}\partial \tau_{1}} u =0;
		\end{equation}
	\item for $k = 5$
		\begin{align}\label{b346}
			b_{3} &=- \frac{1}{2}\frac{\partial^2}{\partial \tau_{1}\partial\tau_{3}} u, &
			b_4 &= \frac{\partial}{\partial \tau_{3}}u, &
			b_{6} &=- 2 \frac{\partial}{\partial \tau_{1}} u \frac{\partial}{\partial \tau_{3}} u - \frac{1}{2}\frac{\partial^{3}}{\partial \tau_{3}\partial \tau_{1}^{2}} u
		\end{align}
	and
		\begin{equation}\label{eq:proofkdv2}
			4 \frac{\partial}{\partial \tau_{1}} u \frac{\partial^{2}}{\partial \tau_{3}\partial \tau_{1}} u + 2 \frac{\partial^{2}}{\partial \tau_{1}^{2}} u \frac{\partial}{\partial \tau_{3}} u - 2 \frac{\partial^{2}}{\partial \tau_{5}\partial \tau_{1}} u + \frac{1}{2}\frac{\partial^{4}}{\partial \tau_{3}\partial \tau_{1}^{3}} u =0;
		\end{equation}
	\item for $k = 7$
		\begin{align}\label{b7810}
			b_7 &= - \frac{1}{2}\frac{\partial^2}{\partial \tau_{1}\partial \tau_{5}} u, &
			b_8 &= \frac{\partial}{\partial \tau_{5}}u, &
			b_{10} &= -2\frac{\partial}{\partial \tau_{5}}u\frac{\partial}{\partial \tau_{1}}u-\frac{1}{2}\frac{\partial^2}{\partial \tau_{1}^2}\frac{\partial}{\partial \tau_{5}}u
		\end{align}
	and
		\begin{equation}\label{eq:proofkdv3}
			4 \frac{\partial}{\partial \tau_{1}} u \frac{\partial^{2}}{\partial \tau_{5}\partial \tau_{1}} u + 2 \frac{\partial^{2}}{\partial \tau_{1}^{2}} u \frac{\partial}{\partial \tau_{5}} u - 2 \frac{\partial^{2}}{\partial \tau_{7}\partial \tau_{1}} u + \frac{1}{2}\frac{\partial^{4}}{\partial \tau_{5}\partial \tau_{1}^{3}} u=0.
	\end{equation}	
\end{itemize}	
Equations \eqref{b012},\eqref{b346},\eqref{b7810}, together with the explicit expressions for $b_5 = b_2 + u_{\tau_3}$ and $b_9 = b_6 + u_{\tau_5}$, are equivalent to the explicit expressions \eqref{eq:Ck}, \eqref{eq:Ak}, \eqref{eq:Dk}, once the potential KdV equations $$u_{\tau_{2k + 1}} = \mathcal{R}_{2k + 1}[u_{\tau_1}], \; k = 1,2,3$$ are proven, which we will do now. The strategy we use is to expand 
$\left(\frac{\partial}{\partial \tau_{2k + 1}}X\right)X^{-1}$ at $\zeta \to \infty$, and set to $0$ the coefficients in $\zeta^{-k}, k = 1,2$ of these expansions. In this way (details are given in Appendix \ref{App:pKdV}) we obtain the equations
\begin{equation}\label{eq:ApppKdV}
	b_2 = \frac{1}{4} b_0^2-2u_{\tau_3}, \quad b_6 =-b_1^2+\frac{b_0^3}{8}-b_0b_4 -2u_{\tau_5} \quad b_{10} = \frac{b_0b_6}{2}-2b_1b_3-b_4b_5 -2u_{\tau_7}
\end{equation}
which, combined with the expressions of $b_2,b_6$ and $b_{10}$ given in \eqref{b012}, \eqref{b346} and \eqref{b7810} give the potential KdV equations \eqref{eq:pKdVhierarchy}. As for the KdV equations \eqref{eq:KdVhierarchy}, the first one is already written in \eqref{eq:proofkdv1}. The other two are proven combining \eqref{eq:proofkdv2} and \eqref{eq:proofkdv3}, with $u_{\tau_{2k + 1}} = \mathcal{R}_{2k + 1}[u_{\tau_1}]$ for $k = 2,3$.
\end{proof}

The following remark will be useful in the subsequent sections, when we will study the asymptotic behavior of the functions $u$ and $v$.

\begin{remark}
As already observed, the function $u$, defined in \eqref{def:u}, can be equivalently rewritten as 
\begin{equation}\label{eq:uvsX}
	u = X_{11}^{(1)} + \i X_{12}^{(1)}= -\frac{\gamma}{s_0^{1/7}} - \frac{h}{s_0^{1/7}}-\left(\frac{s_{1}^{2} t_{1}}{4 s_{0}^{11/7}} - \frac{s_{1} t_{0}}{s_{0}^{6/7}} + \frac{5 s_{1}^{4}}{64 s_{0}^{3}}\right),
\end{equation}
thanks to equation \eqref{eq:logderivativeQ} and Proposition \ref{prop:propertiesX}
Analogously, we also have that $v$, defined in \eqref{def:v} can be written as
\begin{equation}\label{eq:vvsX}
v=-2\i X_{11}^{(1)}X_{12}^{(1)}+2\left(X_{12}^{(1)}\right)^2-2\i X_{12}^{(2)} = \frac{2 \alpha}{s_{0}^{2/7}} - \frac{ \gamma^{2}}{s_{0}^{2/7}} + \frac{y}{s_{0}^{2/7}} - \frac{ s_{1}}{2s_{0}}.
\end{equation}
\end{remark}

\section{The first asymptotic regime}

In this section, we work under the assumption that both Assumptions \ref{assumption1} and \ref{assumption2} hold. We analyze the case in which $\tau_7$ is small and $\tau_5  \geq M (\tau_7)^{5/7}$, for $M$ large enough. In this section and in the ones below, we will sometimes use the variable $s_0$ instead of $\tau_7$. In view of equation \eqref{eq:changeofvar}, this amounts to a simple rescaling, which makes equations clearer. Some conditions must be imposed also on $\tau_1,\tau_3$. Namely, we will work under the following
\begin{assumption}\label{assumptionsfirstregime}
The variables $\tau_1,\tau_3$ satisfy either of the following conditions:
\begin{eqnarray*}
	&\displaystyle \tau_3 \geq \frac{5}4 \frac{\tau_5^2}{s_0} \; \text{and} \; \tau_1 \geq \frac{5 \tau_5^3}{8 s_0^2}, \\
	&\displaystyle \tau_3 \leq \left(\frac{37}{56} - \epsilon \right)\frac{\tau_5^2}{s_0} + \frac{\tau_1 s_0}{\tau_5} \; \text{and} \; \tau_1 \in \mathbb R, \; \text{for any} \; \epsilon>0.
\end{eqnarray*}
\end{assumption} 

The meaning of these assumptions will be clearer in the next subsection. Essentially, they are used to control the behavior of the exponential part in equation \eqref{eq:asympPhi}.

\subsection{Estimates for the jump matrix of $Y$}

We will denote with $J_Y$ the jump matrix of the Riemann-Hilbert problem \ref{RHP2}, explicitly written in equation \eqref{eq:JY}. In the regime of this section, this jump matrix is exponentially close to the identity.

\begin{lemma}\label{lemma:SNcase1}
Let 
\begin{equation}\label{eq:asymptregI}
	\tau_5 \geq M s_0^{5/7}
\end{equation} 
for $M$ some positive real constant, and suppose that Assumptions \eqref{assumptionsfirstregime} are satisfied. Then, there exist two constants $T, k \geq 0$ such that
\begin{equation}\label{eq:estnormFR}
	\norm{J_Y(z)-I}_{L^1(\mathbb{R})} = O(\e^{-k \tau_5/s_0^{5/7}}), \quad \norm{J_Y(z)-I}_{L^2(\mathbb{R})} = O(\e^{-k \tau_5/s_0^{5/7}}), \quad \norm{J_Y(z)-I}_{L^\infty(\mathbb{R})} = O(\e^{-k \tau_5/s_0^{5/7}})
\end{equation}
uniformly for $s_0 \in (0, T]$. In the equation above, $\norm{J_Y(z)-I}_{L^a(\mathbb{R})}, \; a = 1,2,\infty$, denotes the maximum of the entry-wise corresponding $L^a(\mathbb R)$-norm.
\end{lemma}


\begin{proof} Recall, from equation \eqref{eq:JY}, that
$$
J_Y(z)-I=\sigma(\zeta(z))\begin{pmatrix}
\phi_1(z)\phi_2(z) & -\phi_1(z)^2 \vspace{0.3cm}\\ \phi_2(z)^2 & -\phi_1(z)\phi_2(z)
\end{pmatrix},
$$
where $\zeta(z)$ is the same as in Equation \eqref{eq:defzeta}. We will analyze the behavior of $|(J_Y(z))_{21}|$ (the behavior of the other entries is very similar). Let us analyze three different cases:
\begin{enumerate}
	\item Suppose that $z \geq \tau_5/2s_0^{5/7}$. In this regime, using equation \eqref{eq:asympPhi} and the assumptions \ref{assumptionsfirstregime} on $\tau_1$ and $\tau_3$, we obtain
	$$\phi_2(z) = O({\rm e}^{-2 z}) = O({\rm e}^{-z - \tau_5/2s_0^{5/7}})$$ 
	(the exponents are not sharp). Hence, using $\sigma(\zeta(z)) \in [0,1]$, we deduce that
	\begin{equation}\label{est1FR}
		|(J_Y(z))_{21}| = O({\rm e}^{-2z - \tau_5/s_0^{5/7}}).
	\end{equation}
	\item Suppose $|z| < \tau_5/ 2s_0^{5/7}$. In this regime, $\phi_2$ (as well as $\phi_1$) is bounded. On the other hand, 
	$$\zeta(z) = \frac{z}{s_0^{2/7}}- \frac{\tau_5}{s_0} \leq\frac{\tau_5}{2s_0} -\frac{\tau_5}{s_0} = -\frac{\tau_5}{2s_0} \leq -\frac{M}{2}<0.$$
	Take $M>2$, so that $\zeta\leq-1$ and $-|\zeta|^3 \leq-|\zeta|$. Therefore, by Assumption \ref{assumption2},
	$$
		|\sigma(\zeta(z))| \leq k_1 \e^{-k_2|\zeta|^3}\leq k_1 \e^{-k_2|\zeta|} \leq k_1 \e^{-\frac{k_2}{2}|\zeta|}\e^{-k_2\frac{\tau_5}{4s_0}}.
	$$
	Now pick a constant $T \in (0,1)$. For $s_0< T$, one has $s_0 \leq s_0^{5/7}$. Then, for $s_0 \in (0, T]$ we have that
	\begin{equation}\label{est2FR}
		|(J_Y(z))_{21} |= |\sigma(\zeta(z))\phi_2(z)^2|= O(\e^{-k_2\tau_5/4s_0^{5/7}}).
	\end{equation}

	\item At last, we consider the regime in which $z \leq -\tau_5/2s_0^{5/7}$. First notice that $\zeta(z) \leq - 3\tau_5/2s_0<-1$. Therefore
	$\sigma(\zeta(z)) = O({\rm e}^{-k_2| \zeta(z)|}) = O({\rm e}^{2z}{\rm e}^{-\frac{3k_2 \tau_5}{4 s_0}})$. On the other hand, using again equation \eqref{eq:asympPhi}, $\Phi_2(z) = O(z^{1/4})$ so that, 
	\begin{equation}\label{est3FR}
		|(J_Y(z))_{21} |= |\sigma(\zeta(z))\phi_2(z)^2|= O(\e^{z} \e^{-\frac{3k_2 \tau_5}{4 s_0}}) = O(\e^{z} \e^{-\frac{3}4 k_2\frac{\tau_5}{ s_0^{5/7}}}).
	\end{equation}
	where, in the last equality, we used again that, for $s_0 \in (0, T]$, $s_0 \leq s_0^{5/7}$. 
\end{enumerate}

Putting together equations \eqref{est1FR}, \eqref{est2FR} and \eqref{est3FR} we obtain that there exists a positive constant $k$ such that the three estimates in equation \eqref{eq:estnormFR} are satisfied.

\end{proof}

The previous Lemma, using standard Riemann-Hilbert techniques (small norm theorem) implies that, in the regime that we are considering,
$$Y^{(1)} = O (\e^{-k \tau_5/s_0^{5/7}}),
$$
as $\tau_5/s_0^{5/7} \to +\infty$, where $k$ is a positive constant.

\begin{lemma}
Let $\tau_5 \geq M \tau_7^{5/7}$ for $M$ a real constant large enough and suppose that Assumptions \ref{assumptionsfirstregime} are satisfied. Then, the functions $u$ and $v$, defined respectively in \eqref{def:u} and \eqref{def:v}, present the following asymptotic behavior
\begin{align*}
u =& - \left(\frac{2}{7\tau_7}\right)^{1/7}h - \left(\frac{\tau_{1} \tau_{5}}{7\tau_{7}} - \frac{3 \tau_{3} \tau_{5}^{2}}{98\tau_{7}^{2}} + \frac{15 \tau_{5}^{4}}{2744 \tau_{7}^{3}}\right) +O (\e^{-k \tau_5/\tau_7^{5/7}}),\\
v =&\, \left(\frac{2}{7\tau_7}\right)^{2/7}y-\frac{\tau_5}{7\tau_7}+O (\e^{-k \tau_5/\tau_7^{5/7}}),
\end{align*}
as $\tau_7 \to 0$, where $h$ and $y$, defined in \eqref{eq:asympPhi}, are respectively distinguished solutions of the potential KdV equation \eqref{eq:pKdV} and of the $PI^{(2)}$ equation \eqref{eq:PI2}.
\end{lemma}
\begin{proof} We use equations \eqref{eq:uvsX} and \eqref{eq:vvsX}, together with the fact that $\alpha = Y_{11}^{(1)} = O (\e^{-k \tau_5/s_0^{5/7}})$ and the change of variables $(t_0, t_1, s_0, s_1) \mapsto (\tau_1, \tau_3, \tau_5, \tau_7)$, to obtain the claimed result.
\end{proof}

\section{The second asymptotic regime}
In this section, we assume as before that both Assumptions \ref{assumption1} and \ref{assumption2} hold. When convenient, as before, we use the variable $s_0 = \frac{7}2\tau_7$  instead of the latter. We analyze the case in which there exists a constant $M$ such that $|\tau_j | \leq M s_0^{j/7}$, for $j = 1,3,5$.

Let us set 
\begin{equation}\label{eq:defZ}
	Z(\xi) = s_0^{-\sigma_3/14} A^{-1} X(\xi s_0^{-2/7}).
\end{equation} 
It is straightforward to verify that this matrix-valued function solves the following Riemann-Hilbert problem:
\begin{problem}\label{RHPZ}\hfill
\begin{enumerate}
\item $Z(\xi)$ is analytic on $\mathbb{C}\setminus \mathbb R$, and has continuous boundary values $Z_\pm$ on the real line, except for the points $s_0^{2/7}r_1, \cdots, s_0^{2/7}r_k$, satisfying the jump relation
\begin{equation}\label{eq:5.2}
Z_+(\xi)=
Z_-(\xi)\begin{pmatrix}
1 & 1-\sigma(s_0^{-2/7}\xi) \\ 0 & 1
\end{pmatrix},  \xi\in \mathbb R.
\end{equation}
\item As $\xi\to \infty$,
\begin{equation}\label{eq:asympZ}
\begin{split}Z(\xi) &=  \xi^{-\frac{\sigma_3}{4}}N\left[I+\sum_{j\geq 1}\dfrac{Z^{(j)}}{\xi^{j/2}}\right]\e^{-\left(\frac{2}{7}\xi^{\frac{7}{2}}+\kappa_0\xi^{\frac{5}{2}}+\kappa_1\xi^{\frac{3}{2}}+\kappa_2\xi^{\frac{1}{2}}\right)\sigma_3}\\
&\hspace{5cm}
\times \left\lbrace \begin{array}{cc}
I \hspace{0.3cm}&  |\arg \xi|<\pi-\epsilon\\
\begin{pmatrix}
1 & 0 \\ \pm 1 & 1
\end{pmatrix}  \hspace{0.3cm}& \pi-\epsilon<\pm \arg \xi<\pi  
\end{array}\right.,\end{split}
\end{equation}
where $\kappa_0 = {\tau_5}/{s_0^{5/7}}$, $\kappa_1 = {\tau_3}/{s_0^{3/7}}$ and $\kappa_2 = {\tau_1}/{s_0^{1/7}}$.
\end{enumerate}
\end{problem}

Notice that $\kappa_0, \kappa_1, \kappa_2$ are bounded because of our assumptions on the variables $\{\tau_j, \; j = 1,3,5\}$. Thanks to Assumptions \ref{assumption2}, as $s_0 \to 0$ the function $\sigma(s_0^{-2/7}\xi)$ converges pointwise to the indicator function $\iota \chi_{[0,\infty)}$. This suggests to use, as global parametrix away from the origin, the solution $\tilde Z$ of the Riemann-Hilbert problem above, with the jump condition replaced by
\begin{equation}\label{eq:5.2bis}
\tilde{Z}_+(\xi)= 
\tilde{Z}_-(\xi)
\begin{pmatrix}
1 & 1 -\iota \chi_{(0, \infty)}(\xi)\\ 0 & 1
\end{pmatrix}, \qquad \xi \in \mathbb{R}.
\end{equation}
For $\iota = 1$, the Riemann-Hilbert problem \ref{RHPZ} is just a small modification of a Riemann-Hilbert problem considered in \cite{CIK}. Its precise formula will be discussed in the next subsection. For now, the important feature of its solution $\tilde Z$ relies on the fact that, as long as our parameters are bounded, so is $\tilde Z$ for $\xi$ bounded. 

We set the global parametrix to be $\tilde{Z}$ outside the unit circle and, for $|\xi|< 1$, we define a local parametrix by
$$P(\xi) := \tilde{Z}(\xi) \begin{pmatrix}
1&a(\xi) \\ 0& 1
\end{pmatrix},
$$
where, 
$$a(\xi)= \dfrac{1}{2 \pi i}\left(\int_{-\infty}^{0}\dfrac{-\sigma(\xi' s_0^{-2/7})}{\xi'-\xi} \dd \xi'+\int_{0}^{\infty}\dfrac{\iota-\sigma(\xi' s_0^{-2/7})}{\xi'-\xi} \dd \xi' \right), \quad \xi \notin \mathbb{R}.
$$

We now define, for $\xi \notin \mathbb R, \; |\xi| \neq 1$, a new matrix valued function $R(\xi)$ by the equations
\begin{equation}\label{eq:defR}
	R(\xi) := \begin{cases}
			\tilde{A}Z(\xi) \tilde{Z}(\xi)^{-1}, &|\xi|>1, \\
			\\
			\tilde{A}Z(\xi) P(\xi)^{-1}, &|\xi|<1,
		\end{cases}
\end{equation}
	where
$$
	\tilde{A} = \begin{pmatrix}
				1 &0 \\ -\tilde{a} & 1
			\end{pmatrix},
$$
with $\tilde{a} := (-a_{11}+\i a_{12}+\i a_{21}+a_{22})/2$ and $a_{jk} := [Z^{(1)}-\tilde{Z}^{(1)}]_{jk}$.

\begin{lemma}
The matrix-valued function $R(\xi)$, as defined in equation \eqref{eq:defR}, satisfies the Riemann-Hilbert problem \ref{RHPR} detailed below.
\end{lemma}

\begin{problem}\label{RHPR}\hfill
\begin{enumerate}
	\item $R(\xi)$ is analytic for every $\xi \in \mathbb{C}\setminus \Gamma$, where $\Gamma = \{|\xi|=1\} \cup \mathbb{R} \setminus [-1,1]$ as depicted in the Figure \ref{contour:errorcase2}, and has continuous boundary values 	$R_\pm$ satisfying the jump relation
$$
R_+(\xi)=
R_-(\xi)\times \left\lbrace \begin{array}{cc} \tilde{Z}(\xi) P^{-1}(\xi) \hspace{0.3cm}& |\xi|=1 \\ \tilde{Z}_-(\xi) \begin{pmatrix}
1 & \iota\chi_{(0, \infty)}(\xi)-\sigma(\xi s_0^{-2/7}) \\0 & 1 \end{pmatrix}\tilde{Z}_-^{-1}(\xi) \hspace{0.3cm}& \xi \in \mathbb{R}/[-1, 1] \end{array}\right..
$$
	\item As $\xi\to \infty$, with $| \arg(\xi) | < \pi$,
	\begin{align*}
		R(\xi) = I+O(\xi^{-1}).
	\end{align*}
\end{enumerate}
\end{problem}
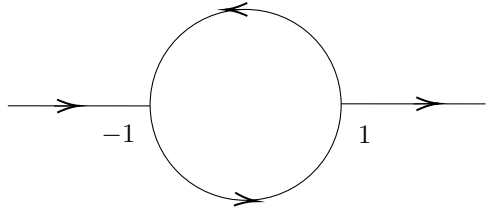
\begin{figure}[H] 
\centering
\begin{tikzpicture}[x=0.75pt,y=0.75pt,yscale=-1,xscale=1]

\draw   (255.66,127.76) .. controls (282.1,127.47) and (303.78,148.67) .. (304.07,175.12) .. controls (304.36,201.57) and (283.16,223.24) .. (256.71,223.53) .. controls (230.26,223.82) and (208.59,202.62) .. (208.3,176.17) .. controls (208.01,149.73) and (229.21,128.05) .. (255.66,127.76) -- cycle ;
\draw    (304.07,175.12) -- (376.44,175.15) ;
\draw    (137,176.11) -- (208.3,176.17) ;
\draw    (255.66,127.76) -- (248.08,127.84) ;
\draw [shift={(246.08,127.86)}, rotate = 359.37] [color={rgb, 255:red, 0; green, 0; blue, 0 }  ][line width=0.75]    (10.93,-3.29) .. controls (6.95,-1.4) and (3.31,-0.3) .. (0,0) .. controls (3.31,0.3) and (6.95,1.4) .. (10.93,3.29)   ;
\draw    (160.2,176.28) -- (170.65,176.16) ;
\draw [shift={(172.65,176.14)}, rotate = 179.37] [color={rgb, 255:red, 0; green, 0; blue, 0 }  ][line width=0.75]    (10.93,-3.29) .. controls (6.95,-1.4) and (3.31,-0.3) .. (0,0) .. controls (3.31,0.3) and (6.95,1.4) .. (10.93,3.29)   ;
\draw    (340.26,175.14) -- (350.71,175.02) ;
\draw [shift={(352.71,175)}, rotate = 179.37] [color={rgb, 255:red, 0; green, 0; blue, 0 }  ][line width=0.75]    (10.93,-3.29) .. controls (6.95,-1.4) and (3.31,-0.3) .. (0,0) .. controls (3.31,0.3) and (6.95,1.4) .. (10.93,3.29)   ;
\draw    (256.71,223.53) -- (259.52,223.81) ;
\draw [shift={(261.51,224)}, rotate = 185.58] [color={rgb, 255:red, 0; green, 0; blue, 0 }  ][line width=0.75]    (10.93,-3.29) .. controls (6.95,-1.4) and (3.31,-0.3) .. (0,0) .. controls (3.31,0.3) and (6.95,1.4) .. (10.93,3.29)   ;

\draw (311.39,185.09) node [anchor=north west][inner sep=0.75pt]    {$1$};
\draw (182.64,184.25) node [anchor=north west][inner sep=0.75pt]    {$-1$};

\end{tikzpicture}

\caption{Contour for the error Riemann-Hilbert problem.}\label{contour:errorcase2} 
\end{figure}

\textbf{Proof:} One can take $s_0$ to be small enough such that $s_0^{2/7}r_j \in (-1,1)$ for all $j \in \{1, \cdots, k\}$. Therefore, the analyticity for every $\xi \not\in \{|\xi|=1\} \cup \mathbb{R} \backslash [-1,1]$ follows by construction. For $\xi \in (-1,1)\backslash\{s_0^{2/7}r_1, \cdots, s_0^{2/7}r_k\}$,we have
\begin{align*}
R_+(\xi) &= \tilde{A} Z_+(\xi) P_+(\xi)^{-1} \\
&= \tilde{A} Z_-(\xi) \begin{pmatrix}
1 & 1-\sigma(\xi s_0^{2/7})\\ 0 & 1
\end{pmatrix}\begin{pmatrix}
1 & -a_+(\xi)\\ 0 & 1
\end{pmatrix}\begin{pmatrix}
1 & \iota \chi_{[0,\infty)}-1\\ 0 & 1
\end{pmatrix}\begin{pmatrix}
1 & a_-(\xi)\\ 0 & 1
\end{pmatrix} P_-(\xi)^{-1} \\
&= \tilde{A} Z_-(\xi) \begin{pmatrix}
1 & (a_-(\xi)-a_+(\xi))-\sigma(\xi s_0^{2/7})+\iota \chi_{[0,\infty)}\\ 0 & 1
\end{pmatrix} P_-(\xi)^{-1}. 
\end{align*}
By definition, $a(\xi)$ is the Cauchy transform of $(\iota\chi_{(0, \infty)}(\xi)-\sigma(\xi s_0^{-2/7}))$. Therefore, by the Plemelj formula $$a_+(\xi)-a_-(\xi) = \iota\chi_{(0, \infty)}(\xi)-\sigma(\xi s_0^{-2/7}).$$ Consequently, 
\begin{align*}
R_+(\xi) &= \tilde{A} Z_-(\xi) P_-(\xi)^{-1} = R_-(\xi), 
\end{align*}
and $R(\xi)$ is analytic for $\xi \in (-1,1)\backslash\{s_0^{2/7}r_1, \cdots, s_0^{2/7}r_k\}$. We need to be more careful around the points $s_0^{2/7}r_j$. Set $\ell_j = \lim_{\epsilon \downarrow 0}(\sigma(r_j+\epsilon)-\sigma(r_j-\epsilon))$. The limit exists by Assumption \ref{assumption1}. In a neighborhood of $s_0^{2/7}r_j$ we represent $Z(\xi)$ as follows
$$Z(\xi) = E_j(\xi) \begin{pmatrix}
1 & \frac{\ell_j}{2\pi i} \log(s_0^{-2/7}\xi-r_j) \\ 0 &1
\end{pmatrix},
$$
where we take the principal branch of the logarithm. Notice that
\begin{align*}
E_{j,+}(\xi) &= Z_+(\xi) \begin{pmatrix}
1 & -\frac{\ell_j}{2\pi i} \log(s_0^{-2/7}\xi-r_j)|_+ \\ 0 &1
\end{pmatrix}\\
&=E_{j,-}(\xi)\begin{pmatrix}
1 & \frac{\ell_j}{2\pi i} \log(s_0^{-2/7}\xi-r_j)|_- \\ 0 &1
\end{pmatrix} \begin{pmatrix}
1 & 1-\sigma(s_0^{-2/7}\xi) \\ 0 & 1
\end{pmatrix} \begin{pmatrix}
1 & -\frac{\ell_j}{2\pi i} \log(s_0^{-2/7}\xi-r_j)|_+ \\ 0 &1
\end{pmatrix}\\
&=E_{j,-}(\xi)\begin{pmatrix}
1 & 1-(\sigma(s_0^{-2/7}\xi)+\ell_j \chi_{(-1, s_0^{2/7}r_j)}(\xi)) \\ 0 &1
\end{pmatrix}.
\end{align*}
Therefore, the jump of $E_j$ is continuous on $s_0^{2/7}r_j$. With this representation, we can now evaluate the jump of $R(\xi)$ in the neighbourhood of such points as follows
\begin{align*}
R_+(\xi) &= \tilde{A}Z_{+}(\xi)P_{+}(\xi)^{-1} \\
&= \tilde{A}E_{j,+}(\xi) \begin{pmatrix}
1 & \frac{\ell_j}{2\pi i} \log(s_0^{-2/7}\xi-r_j)|_+ \\ 0 &1
\end{pmatrix}P_{+}(\xi)^{-1} \\
&= \tilde{A}E_{j,-}(\xi) \begin{pmatrix}
1 & \frac{\ell_j}{2\pi i} \log(s_0^{-2/7}\xi-r_j)|_- \\ 0 &1
\end{pmatrix}\begin{pmatrix}
1 & \star \\ 0 &1
\end{pmatrix}P_{-}(\xi)^{-1}\\
\star &= -\frac{\ell_j}{2\pi i} \log(s_0^{-2/7}\xi-r_j)|_-+\frac{\ell_j}{2\pi i} \log(s_0^{-2/7}\xi-r_j)|_+ -\ell_j-\sigma(\xi s_0^{2/7})+a_-(\xi)-a_+(\xi)+\iota \chi_{(0, \infty)}(\xi).
\end{align*}
Since $a_+(\xi)-a_-(\xi) = \iota\chi_{(0, \infty)}(\xi)-\sigma(\xi s_0^{-2/7})$ and $\frac{\ell_j}{2\pi i} \log(s_0^{-2/7}\xi-r_j)|_+ -\frac{\ell_j}{2\pi i} \log(s_0^{-2/7}\xi-r_j)|_- = \ell_j$ we conclude that $\star = 0$, and
\begin{align*}
R_+(\xi) &= \tilde{A}E_{j,-}(\xi) \begin{pmatrix}
1 & \frac{\ell_j}{2\pi i} \log(s_0^{-2/7}\xi-r_j)|_- \\ 0 &1
\end{pmatrix}P_{-}(\xi)^{-1} = R_-(\xi).
\end{align*}
Consequently, $R(\xi)$ is analytic for $\xi \in (-1,1)$. For $\xi \in \mathbb{R}/[-1,1]$,
\begin{align*}
R_+(\xi) =& \tilde{A}Z_+(\xi) \tilde{Z}^{-1} _+(\xi) = \tilde{A}Z_-(\xi) \tilde{Z}_-(\xi)^{-1}\tilde{Z}_-(\xi)J_Z(\xi) J_{\tilde{Z}}^{-1}(\xi)\tilde{Z}^{-1}_-(\xi)\\
=& R_-(\xi)\tilde{Z}_-(\xi)J_Z(\xi) J_{\tilde{Z}}^{-1}(\xi)\tilde{Z}^{-1}_-(\xi),
\end{align*}
which gives the correct expression for the jump. For $|\xi|=1$,
\begin{align*}
R_+(\xi) =& \tilde{A}Z(\xi)P^{-1}(\xi)\\
=& \tilde{A}Z(\xi)\tilde{Z}^{-1}(\xi)\tilde{Z}(\xi)P^{-1}(\xi) = R_-(\xi)\tilde{Z}(\xi)P^{-1}(\xi),
\end{align*}
and $J_R = \tilde{Z}(\xi)P(\xi)^{-1}$. 
For the asymptotic condition, notice that as $\xi \to \infty$, $|\arg \xi|<\pi$,
\begin{align*}
Z(\xi)\tilde{Z}(\xi)^{-1} &=  \xi^{-\frac{\sigma_3}{4}}N\left[I+\dfrac{Z^{(1)}}{\xi^{1/2}}+\dfrac{Z^{(2)}}{\xi}+O(\xi^{-3/2})\right]\left[I-\dfrac{\tilde{Z}^{(1)}}{\xi^{1/2}}+\dfrac{\tilde{Z}^{(-2)}}{\xi}+O(\xi^{-3/2})\right]N^{-1}\xi^{\frac{\sigma_3}{4}}\\
&=  \xi^{-\frac{\sigma_3}{4}}N\left[I+\dfrac{Z^{(1)}-\tilde{Z}^{(1)}}{\xi^{1/2}}+\dfrac{Z^{(2)}+\tilde{Z}^{(-2)}-Z^{(1)}\tilde{Z}^{(1)}}{\xi}+O(\xi^{-3/2})\right]N^{-1}\xi^{\frac{\sigma_3}{4}},
\end{align*}
where $\tilde{Z}^{(-2)} = -\tilde{Z}^{(1)}\tilde{Z}^{(-1)}-\tilde{Z}^{(2)}$\footnote{The expansion follows from the fact that $\sum_{j\geq 0} \tilde{Z}^{(j)}/\xi^{j/2}\left(\sum_{j\geq 0} \tilde{Z}^{(j)}/\xi^{j/2}\right)^{-1} = I$. Then, denoting $\left(\sum_{j\geq 0} \tilde{Z}^{(j)}/\xi^{j/2}\right)^{-1}=:\sum_{j\geq 0} \tilde{Z}^{(-j)}/\xi^{j/2}$, expanding and collecting the terms of same order, one has $\tilde{Z}^{(-1)}=-\tilde{Z}^{(1)}$ and $\tilde{Z}^{(-2)}=-\tilde{Z}^{(-1)}\tilde{Z}^{(1)}-\tilde{Z}^{(2)}$.}. This gives
\begin{align*}
Z(\xi)\tilde{Z}(\xi)^{-1} &=  \begin{pmatrix}
1 & 0 \\ \frac{-a_{11}+\i a_{12}+\i a_{21}+a_{22}}{2} & 1
\end{pmatrix}+ \frac{1}{\xi^{1/2}}\begin{pmatrix}
\frac{a_{11}-\i a_{12}+\i a_{21}+a_{22}}{2} &0 \\
\frac{-b_{11}+ib_{12}+\i b_{21}+b_{22}}{2} & \frac{a_{11}+\i a_{12}-\i a_{21}+a_{22}}{2}
\end{pmatrix}+O(\xi^{-1}),
\end{align*}
where $a_{jk} = [Z^{(1)}-\tilde{Z}^{(1)}]_{jk}$ and $b_{jk} = [Z^{(2)}-\tilde{Z}^{(2)}- Z^{(1)}\tilde{Z}^{(1)}-\tilde{Z}^{(1)}\tilde{Z}^{(-1)}]_{jk}$. The symmetry in the asymptotic expansions implies that $a_{11}=-a_{22}$, $a_{12}=a_{21}$, $b_{11}=b_{22}$ and $b_{12}=-b_{21}$. 
The matrix of order $\xi^{1/2}$ is then equal to zero and we are left with
\begin{align*}
Z(\xi)\tilde{Z}(\xi)^{-1} &= \begin{pmatrix}
1 &0 \\ \tilde{a} & 1
\end{pmatrix}+ O(\xi^{-1}),
\end{align*}
where $\tilde{a} = (-a_{11}+\i a_{12}+\i a_{21}+a_{22})/2$. Therefore,
\begin{align*}
R(\xi) =&\begin{pmatrix}
1 &0 \\ -\tilde{a} & 1
\end{pmatrix} Z(\xi)\tilde{Z}(\xi)^{-1} = I+ O(\xi^{-1}),
\end{align*}
as claimed.

 \qed

\begin{lemma}
Let $|\tau_j| \leq M s_0^{j/7}, \; j = 1,3,5$ for $M>0$ some real constant large enough. Then, there exists $\epsilon >0$ such that
\begin{equation}\label{eq:estnormJR}
	\norm{J_R-I}_{L^1(\Gamma)} = O(s_0^{2/7}), \quad \norm{J_R-I}_{L^2(\Gamma)} = O(s_0^{2/7}), \quad \norm{J_R-I}_{L^\infty(\Gamma)} = O(s_0^{2/7})
\end{equation}
uniformly for $s_0 \in (0, \epsilon)$. In the equation above, $\norm{J_R-I}_{L^a(\Gamma)}, \; a = 1,2,\infty$, denotes the maximum of the entry-wise corresponding $L^a(\Gamma)$-norm.
\end{lemma}

\textbf{Proof:} Notice that for $\xi \in \mathbb{R}/[-1,1]$,
\begin{align*}
\norm{J_R(\xi)-I} = & \norm{(\iota\chi_{(0, \infty)}(\xi)-\sigma(\xi s_0^{-2/7})) \tilde{Z}_-(\xi)E_{12}\tilde{Z}_-^{-1}(\xi)} \\
= &\, O(\e^{-c_2s_0^{-2/7}|\xi|^3-\frac{4}{7}\xi^{\frac{7}{2}}+2\kappa_0\xi^{\frac{5}{2}}-2\kappa_1\xi^{\frac{3}{2}}+2\kappa_2\xi^{\frac{1}{2}}}|\xi|^{1/2}) = o(s_0^{2/7}).
\end{align*}

On the unit circle,
\begin{align*}
J_R(\xi) - I = & \tilde{Z}(\xi) \begin{pmatrix}
0 & -a(\xi) \\ 0 & 0
\end{pmatrix}\tilde{Z}^{-1}(\xi),
\end{align*}
and $\tilde Z(\xi)$ is bounded for $\xi$ bounded. Thus, all we need to do is to study the behavior of $a(\xi)$ as $s_0 \to 0$.

To this aim, we start with a change in variables $\zeta' = \xi's_0^{-2/7}$. Therefore,
$$a(\xi)= \dfrac{s_0^{2/7}}{2 \pi i}\int_{-\infty}^{\infty}\dfrac{\iota\chi_{(0, \infty)} (\zeta')-\sigma(\zeta')}{\zeta's_0^{2/7}-\xi} \dd \zeta'.
$$

If $|\Im \xi| \geq \delta$ for some $\delta \in \mathbb{R}$ positive, then $|\zeta's_0^{2/7}-\xi|^{-1} \leq |\Im \xi|^{-1} \leq \delta^{-1}$. Therefore, uniformly for $|\xi|=1$ and $|\Im \xi| \geq \delta$,
\begin{align*}
|a(\xi)| \leq & \dfrac{s_0^{2/7}}{2 \pi}\int_{-\infty}^{\infty} \left|\dfrac{\iota\chi_{(0, \infty)} (\zeta')-\sigma(\zeta')}{\zeta's_0^{2/7}-\xi}\right| \dd \zeta'\\
\leq & \dfrac{s_0^{2/7}}{2 \pi \delta}\int_{-\infty}^{\infty} \left|\iota\chi_{(0, \infty)}(\zeta')-\sigma(\zeta')\right| \dd \zeta' \leq \dfrac{2s_0^{2/7}k_1}{\delta k_2}.
\end{align*}

Now we need to analyze the neighborhoods of $1$ and $-1$. By symmetry, it suffices to analyze one of the cases. Take, for instance $\xi=1$ and consider the compact interval $K=[s_0^{-2/7}/2, 3s_0^{-2/7}/2]$. The integral can be expressed as
\begin{align*}
a(\xi)= \dfrac{s_0^{2/7}}{2 \pi i} \left[\int_{\mathbb{R}/K}\dfrac{\iota\chi_{(0, \infty)} (\zeta')-\sigma(\zeta')}{\zeta's_0^{2/7}-\xi} \dd \zeta' + \int_{K}\dfrac{1-\sigma(\zeta')}{\zeta's_0^{2/7}-\xi} \dd \zeta' \right].
\end{align*}

For some $k_0 \geq 0$, it follows that $|\zeta's_0^{2/7}-\xi| \geq k_0$ for all $\zeta' \in \mathbb{R}/K$. Thus
\begin{align*}
\left|\int_{\mathbb{R}/K}\dfrac{\iota\chi_{(0, \infty)} (\zeta')-\sigma(\zeta')}{\zeta's_0^{2/7}-\xi} \dd \zeta'\right| \leq & \dfrac{1}{k_0} \int_{\mathbb{R}/K}|\iota\chi_{(0, \infty)} (\zeta')-\sigma(\zeta')| \dd \zeta'\\
\leq & \dfrac{1}{k_0} \int_{\mathbb{R}}|\iota\chi_{(0, \infty)} (\zeta')-\sigma(\zeta')| \dd \zeta' \leq \dfrac{2 k_1}{k_0 k_2}.
\end{align*}

For the other integral, with $|\xi|=1$, notice that $\zeta' \in K$ implies $\zeta' s_0^{2/7} \in [\frac{1}{2}, \frac{3}{2}]$. Moreover,
$$ |\zeta's_0^{2/7}-\xi| \geq |\zeta's_0^{2/7}-1| \implies \frac{1}{|\zeta's_0^{2/7}-\xi|} \leq \frac{1}{|\zeta's_0^{2/7}-1|}.
$$

Now consider
\begin{align*}
\left|\int_{K}\dfrac{\iota-\sigma(\zeta')}{\zeta's_0^{2/7}-\xi} \dd \zeta'\right| \leq& \underbrace{\left|\int_{K}\dfrac{\iota-\sigma(s_0^{-2/7})}{\zeta's_0^{2/7}-\xi} \dd \zeta'\right|}_{I_1}+\underbrace{\left|\int_{K}\dfrac{\sigma(s_0^{-2/7})-\sigma(\zeta')}{\zeta's_0^{2/7}-1} \dd \zeta'\right|}_{I_2}.
\end{align*}

For the second integral, we apply the mean value theorem, obtaining
\begin{align*}
I_2 =& \frac{1}{s_0^{2/7}}\int_{s_0^{-2/7}/2}^{3s_0^{-2/7}/2}\left|\dfrac{\sigma(\zeta')-\sigma(s_0^{-2/7})}{\zeta'-s_0^{-2/7}}\right| \dd \zeta'\\
\leq & \frac{k_3}{s_0^{2/7}}\int_{s_0^{-2/7}/2}^{3s_0^{-2/7}/2}\sup_{z \in K}z^{-2} \dd \zeta' = \frac{4 k_3}{9}.
\end{align*}

On the other side,
\begin{align*}
I_1 = \frac{1}{s_0^{2/7}}\left| \int_{1/2}^{3/2}\dfrac{\iota-\sigma(s_0^{-2/7})}{\zeta'-\xi}\dd \zeta' \right| \leq \frac{1}{s_0^{2/7}} \e^{-k_2 s_0^{-2/7}}\left| \int_{1/2}^{3/2}\dfrac{1}{\zeta'-\xi}\dd \zeta' \right| = o(1).
\end{align*}

Altogether we conclude that $|a(\xi)| = O(s_0^{2/7})$ and, therefore, equations \eqref{eq:estnormJR} are satisfied. \qed

\subsection{Analysis of the global parametrix}

The Riemann-Hilbert problem for $\tilde Z$, defined by equations \eqref{eq:5.2bis} and \eqref{eq:asympZ}, is closely related to the Riemann-Hilbert problem associated to the third member of the Painlevé II hierarchy, in a way which will be presented in this subsection. Let us start considering the following

\begin{problem}\label{RHPP2}\hfill
\begin{enumerate}
\item $\Psi(\zeta;x_0,x_1,x_2) \equiv \Psi(\zeta)$ is analytic for every $\zeta \in \mathbb{C}\setminus  \Gamma$, where 
$$\Gamma= \cup_{j=1}^6 \Gamma_j, \quad \text{for} \quad \Gamma_1=\mathbb{R_+}\e^{\frac{\pi\i}{14}}, \; \Gamma_2 = \mathbb{R_+}\e^{\frac{\pi\i}{2}},  \; \Gamma_3 = \mathbb{R_+}\e^{\frac{6\pi\i}{14}} \quad \text{and} \quad \Gamma_{k + 3} = -\Gamma_k, \; k = 1,2,3.$$
(see Figure \ref{contourPII3}), and has continuous boundary values $\Psi_{\pm}$ satisfying the jump relation
$$\Psi_{+}(\zeta)=
\Psi_{-}(\zeta)J_k, \quad z \in \Gamma_k,$$ where
\begin{equation}
J_1 = \begin{pmatrix}
1 & 0 \\ -\i & 1
\end{pmatrix}, \qquad
J_2 = \begin{pmatrix}
1 & s \\ 0 & 1
\end{pmatrix}, \qquad 
J_3 = \begin{pmatrix}
1 & 0 \\ -\i & 1
\end{pmatrix},
\end{equation}
and $J_{k+3}=J_k^T$.
\item As $\zeta\to \infty$,
$$
\Psi(\zeta)=\left(I+\dfrac{1}{\zeta}\Psi^{(1)}+O(\zeta^{-2})\right)\e^{-i \Theta(\zeta) \sigma_3},
$$
where $\Theta(\zeta) =  \frac{2^6}{7}\zeta^7 -\frac{2^4}{5}x_2 \zeta^5+\frac{4}{3}x_1 \zeta^3+x_0\zeta$.
\item As $\zeta \to 0$,
\begin{align*}
\Psi(\zeta)= \left\lbrace \begin{array}{ll}
O\begin{pmatrix}
|\zeta|^{1/2} & |\zeta|^{-1/2} \\
|\zeta|^{1/2} & |\zeta|^{-1/2}
\end{pmatrix} \begin{pmatrix}
1 &0 \\  \i & 1
\end{pmatrix},  \hspace{0.3cm}& \zeta \in S_1, \\ 
O\begin{pmatrix}
|\zeta|^{1/2} & |\zeta|^{-1/2} \\
|\zeta|^{1/2} & |\zeta|^{-1/2}
\end{pmatrix} \begin{pmatrix}
1 &0 \\ - \i & 1
\end{pmatrix},  \hspace{0.3cm}& \zeta \in S_4 \\ 
O\begin{pmatrix}
|\zeta|^{1/2} & |\zeta|^{-1/2} \\
|\zeta|^{1/2} & |\zeta|^{-1/2}
\end{pmatrix} &  \zeta \in S_2\cup S_3\\
O\begin{pmatrix}
|\zeta|^{-1/2} & |\zeta|^{1/2} \\
|\zeta|^{-1/2} & |\zeta|^{1/2}
\end{pmatrix} &  \zeta \in S_5 \cup S_6.
\end{array}\right.
\end{align*}
\end{enumerate}
\end{problem}

\begin{figure}[h]
\centering

\tikzset{every picture/.style={line width=0.75pt}} 

\begin{tikzpicture}[x=0.75pt,y=0.75pt,yscale=-1,xscale=1]

\draw    (188.88,168.48) -- (266.49,149.06) ;
\draw [shift={(268.43,148.57)}, rotate = 165.95] [color={rgb, 255:red, 0; green, 0; blue, 0 }  ][line width=0.75]    (10.93,-3.29) .. controls (6.95,-1.4) and (3.31,-0.3) .. (0,0) .. controls (3.31,0.3) and (6.95,1.4) .. (10.93,3.29)   ;
\draw    (326,136.56) -- (268.43,148.57) ;

\draw    (191.62,166.83) -- (117.08,147.73) ;
\draw [shift={(115.14,147.23)}, rotate = 14.37] [color={rgb, 255:red, 0; green, 0; blue, 0 }  ][line width=0.75]    (10.93,-3.29) .. controls (6.95,-1.4) and (3.31,-0.3) .. (0,0) .. controls (3.31,0.3) and (6.95,1.4) .. (10.93,3.29)   ;
\draw    (115.14,147.23) -- (59,134.56) ;

\draw  [fill={rgb, 255:red, 0; green, 0; blue, 0 }  ,fill opacity=1 ] (191.14,164.84) .. controls (192.27,164.84) and (193.18,165.9) .. (193.19,167.22) .. controls (193.19,168.53) and (192.28,169.6) .. (191.16,169.6) .. controls (190.04,169.61) and (189.12,168.54) .. (189.12,167.23) .. controls (189.12,165.92) and (190.02,164.85) .. (191.14,164.84) -- cycle ;
\draw    (191.33,167.2) -- (266.44,184.63) ;
\draw [shift={(268.39,185.08)}, rotate = 193.06] [color={rgb, 255:red, 0; green, 0; blue, 0 }  ][line width=0.75]    (10.93,-3.29) .. controls (6.95,-1.4) and (3.31,-0.3) .. (0,0) .. controls (3.31,0.3) and (6.95,1.4) .. (10.93,3.29)   ;
\draw    (268.39,185.08) -- (325,196.56) ;

\draw    (192.43,167.69) -- (118.74,183.52) ;
\draw [shift={(116.78,183.94)}, rotate = 347.88] [color={rgb, 255:red, 0; green, 0; blue, 0 }  ][line width=0.75]    (10.93,-3.29) .. controls (6.95,-1.4) and (3.31,-0.3) .. (0,0) .. controls (3.31,0.3) and (6.95,1.4) .. (10.93,3.29)   ;
\draw    (62,193.56) -- (116.78,183.94) ;

\draw  [dash pattern={on 0.84pt off 2.51pt}]  (63,165.56) -- (328,168.56) ;
\draw    (191.55,167.46) -- (191.31,218.47) ;
\draw [shift={(191.3,220.47)}, rotate = 270.27] [color={rgb, 255:red, 0; green, 0; blue, 0 }  ][line width=0.75]    (10.93,-3.29) .. controls (6.95,-1.4) and (3.31,-0.3) .. (0,0) .. controls (3.31,0.3) and (6.95,1.4) .. (10.93,3.29)   ;
\draw    (191.3,220.47) -- (192.2,259.16) ;

\draw    (191.8,166.16) -- (192.49,115.16) ;
\draw [shift={(192.51,113.16)}, rotate = 90.77] [color={rgb, 255:red, 0; green, 0; blue, 0 }  ][line width=0.75]    (10.93,-3.29) .. controls (6.95,-1.4) and (3.31,-0.3) .. (0,0) .. controls (3.31,0.3) and (6.95,1.4) .. (10.93,3.29)   ;
\draw    (192.51,113.16) -- (191.95,74.46) ;

\draw (55,95.96) node [anchor=north west][inner sep=0.75pt]  [font=\footnotesize]  {$\begin{pmatrix}
1 & 0\\
-i & 1
\end{pmatrix}$};
\draw (289,99.96) node [anchor=north west][inner sep=0.75pt]  [font=\footnotesize]  {$\begin{pmatrix}
1 & 0\\
-i & 1
\end{pmatrix}$};
\draw (282,201.96) node [anchor=north west][inner sep=0.75pt]  [font=\footnotesize]  {$\begin{pmatrix}
1 & -i\\
0 & 1
\end{pmatrix}$};
\draw (56,200.96) node [anchor=north west][inner sep=0.75pt]  [font=\footnotesize]  {$\begin{pmatrix}
1 & -i\\
0 & 1
\end{pmatrix}$};
\draw (230,121.96) node [anchor=north west][inner sep=0.75pt]    {$S_{2}$};
\draw (151,119.96) node [anchor=north west][inner sep=0.75pt]    {$S_{3}$};
\draw (307,147.96) node [anchor=north west][inner sep=0.75pt]    {$S_{1}$};
\draw (73,146.96) node [anchor=north west][inner sep=0.75pt]    {$S_{4}$};
\draw (72,169.46) node [anchor=north west][inner sep=0.75pt]    {$S_{5}$};
\draw (149,196.46) node [anchor=north west][inner sep=0.75pt]    {$S_{6}$};
\draw (193.95,77.86) node [anchor=north west][inner sep=0.75pt]  [font=\footnotesize]  {$\begin{pmatrix}
1 & s\\
0 & 1
\end{pmatrix}$};
\draw (196.95,219.86) node [anchor=north west][inner sep=0.75pt]  [font=\footnotesize]  {$\begin{pmatrix}
1 & 0\\
s & 1
\end{pmatrix}$};
\draw (222.5,196.46) node [anchor=north west][inner sep=0.75pt]    {$S_{7}$};
\draw (308,171.96) node [anchor=north west][inner sep=0.75pt]    {$S_{8}$};

\end{tikzpicture}

\caption{Contour and jumps for the $P_{II}^{(3)}$ Painlevé hierarchy RHP.}
\label{contourPII3}
\end{figure}
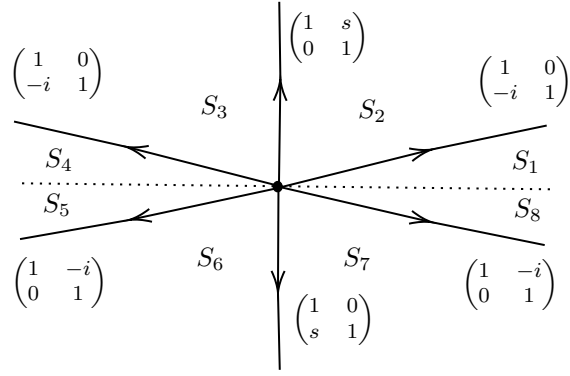
 
 This Riemann-Hilbert problem, depending parametrically on $s$, is a specific case of the general Riemann-Hilbert problem associated to the third member of the Painlev\'e II hierarchy, which depends on seven parameters (the Stokes parameters) $s_i,\, i = 1,\ldots,7$ (see for instance \cite{CIK} and references therein). In the case of interest for our paper, we set $s_1 = s_7 = -\i$, and $s_4 = s$.  Using standard methods about Riemann-Hilbert formulation of isomonodromic problems, one can prove (see details in \cite{CIK}) that 
\begin{equation}\label{eq:pandq}
	\Psi^{(1)} = \Psi^{(1)}(x_0,x_1,x_2) = \begin{pmatrix}
										\i p(x_0,x_1,x_2) & -\frac{\i}2 q(x_0,x_1,x_2) \\
										\frac{\i}2 q(x_0,x_1,x_2) & - \i p(x_0,x_1,x_2)
					\end{pmatrix},
\end{equation}
where $p = p(x_0,x_1,x_2)$ and $q = q(x_0,x_1,x_2)$ are real-valued functions such that 
$$
	\frac{\partial p}{\partial x_0} = \frac{1}2 q^2,
$$
and, in addition, $q$ solves the following equations:
\begin{align}
	\frac{\partial}{\partial x_{1}} q =& 2 q^2 \frac{\partial}{\partial x_{0}} q- \frac{1}{3} \frac{\partial^{3}}{\partial x_{0}^{3}} q, \label{mKdV1} \\
	\frac{\partial}{\partial x_{2}} q =& - 6 q^{4} \frac{\partial}{\partial x_{0}} q + 2 q^{2} \frac{\partial^{3}}{\partial x_{0}^{3}} q + 8 q \frac{\partial}{\partial x_{0}} q \frac{\partial^{2}}{\partial x_{0}^{2}} q + 2 \left(\frac{\partial}{\partial x_{0}} q\right)^{3} - \frac{1}{5} \frac{\partial^{5}}{\partial x_{0}^{5}} q, \label{mKdV2}\\ \nonumber
	 \frac{\partial^{6}}{\partial x_{0}^{6}} q =& \, x_{0} q + x_{1} \left(2q^{3} -  \frac{\partial^{2}}{\partial x_{0}^{2}} q \right) +x_2\left(-6 q^{5} + 10 q^{2} \frac{\partial^{2}}{\partial x_{0}^{2}} q + 10 q \frac{\partial}{\partial x_{0}} q^{2}  - \frac{\partial^{4}}{\partial x_{0}^{4}} q\right)  \\ \nonumber
&+ 20 q^{7} - 70 q^{4} \frac{\partial^{2}}{\partial x_{0}^{2}} q - 140 q^{3} \left(\frac{\partial}{\partial x_{0}} q\right)^{2} + 14 q^{2} \frac{\partial^{4}}{\partial x_{0}^{4}} q + 56 q \frac{\partial}{\partial x_{0}} q \frac{\partial^{3}}{\partial x_{0}^{3}} q  \\
&+ 42 q \left(\frac{\partial^{2}}{\partial x_{0}^{2}} q\right)^{2} + 70 \left(\frac{\partial}{\partial x_{0}} q\right)^{2} \frac{\partial^{2}}{\partial x_{0}^{2}} q  -\frac{1}2. \label{PII3}
\end{align}
Equations \eqref{mKdV1} and \eqref{mKdV2} are the two first members of the modified KdV hierarchy, while equation \eqref{PII3} is the third member of the Painlevé II hierarchy. We now define, as in \cite{CIK}, the following folding transformation:
\begin{align*}
\Xi(x_0,x_1,x_2;z) \equiv {\Xi}(z)&= z^{-\sigma_3/4}N \e^{\pi \i\sigma_3/4}\Psi(\i z^{1/2})\e^{-\pi \i\sigma_3/4}.
\end{align*}

The matrix-valued function $\Xi$ will then have jumps on the oriented contours $$\tilde\Gamma := \mathbb R \cup \mathbb{R_+}\e^{\pm\frac{6\pi\i}{7}}$$ and solves the following Riemann-Hilbert problem
\begin{problem}\label{RHPP34}\hfill
\begin{enumerate}
\item $\Xi(z)$ is analytic on $\mathbb{C}\setminus \tilde\Gamma$ and has continuous boundary values $\Xi$ satisfying the jump relation
\begin{equation}
\Xi_+(z)=
\Xi_-(z)\times \left\lbrace \begin{array}{ll}\begin{pmatrix}
1 & 0 \\ 1 & 1
\end{pmatrix}  \hspace{0.3cm}& z \in \mathbb{R_+}\e^{\pm\frac{6\pi\i}{7}} \vspace{0.3cm}\\
\begin{pmatrix}
0 & 1 \\ -1 & 0
\end{pmatrix}  \hspace{0.3cm}& z \in (-\infty,0)\vspace{0.3cm}\\
\begin{pmatrix}
1 & \i s \\ 0 & 1
\end{pmatrix}  \hspace{0.3cm}& z \in (0, \infty)
\end{array}\right..
\end{equation}
\item As $z\to \infty$,
\begin{equation} \label{eq:asympXi}
\Xi(z) =  z^{-\frac{\sigma_3}{4}}N\left[I+\sum_{j\geq 1}\dfrac{\e^{\pi \i\sigma_3/4}\Psi^{(j)}\e^{-\pi \i\sigma_3/4}}{z^{j/2}}\right]\e^{-\tilde{\theta}(z)\sigma_3}
\end{equation}
where $\tilde{\theta}(z) = \frac{2^6}7 z^{7/2}+\frac{2^4}5 x_2z^{5/2}+\frac{4}3 x_1z^{3/2}-x_0z^{1/2}$.
\item As $z \to 0$,
$$\Xi(z) = O(|z|^{-1/2}).$$
\end{enumerate}
\end{problem}

In the case $s = 0$, the existence of the solution of the above Riemann-Hilbert problem had been proven in \cite{CIK}. We will extend in a moment the same result for arbitrary $s \in \i[0,1]$. This is important for us since, as it is easy to check, the Riemann-Hilbert problem \ref{RHPP34} is closely related to the global parametrix $\tilde Z$ satisfying the Riemann-Hilbert problem with conditions \eqref{eq:5.2bis}, \eqref{eq:asympZ}. Namely, we have the following
\begin{proposition}\label{relationZXi}
	Suppose that a solution $\Xi$ of the Riemann-Hilbert problem \ref{RHPP34} exists. Then, 
	$$\tilde Z(\xi)  = c^{\sigma_3/4}\Xi(c\xi;x_0,x_1,x_2),$$
	where $x_0 = -2^{5/7} \kappa_2, \, x_1 = 3\cdot 2^{1/7}\kappa_1, x_2 = 5 \cdot 2^{-3/7} \kappa_0$ and $c = 2^{-10/7}$. Moreover, the Stokes parameters $\iota$ and $s$ are related by $s = \i (\iota - 1)$.
\end{proposition}


We now prove the existence of $\Xi$ by the following vanishing lemma
\begin{lemma}[Vanishing Lemma] Let $x_0, x_1, x_2 \in \mathbb{R}$ and consider the matrix-valued function $\hat{\Psi}(\zeta)$ that satisfies the Riemann-Hilbert problem \ref{RHPP2}, but with asymptotic condition given by
$$\hat{\Psi}(\zeta) \e^{\i \Theta(\zeta)\sigma_3} = O(1/\zeta), \qquad \text{as } \zeta \to \infty.
$$
Then, $\hat{\Psi}(\zeta) \equiv 0$.
\end{lemma}
\textbf{Proof:} We start with the definition of an auxiliary matrix-function
\begin{equation}
A(\zeta) = \left\lbrace \begin{array}{ll}
\hat{\Psi}(\zeta) \e^{\i \Theta(\zeta) \sigma_3}, & \text{for } \zeta \in S_2 \cup S_3 \cup S_6 \cup S_7, \\
\hat{\Psi}(\zeta) \begin{pmatrix}
1 &0 \\ -\i &1
\end{pmatrix}\e^{\i \Theta(\zeta) \sigma_3}, & \text{for } \zeta \in S_1, \\
\hat{\Psi}(z) \begin{pmatrix}
1 &-\i \\ 0&1
\end{pmatrix}\e^{\i \Theta(\zeta) \sigma_3}, & \text{for } \zeta \in S_5, \\
\hat{\Psi}(\zeta) \begin{pmatrix}
1 &0 \\ \i &1
\end{pmatrix}\e^{\i \Theta(\zeta) \sigma_3}, & \text{for } \zeta \in S_4, \\
\hat{\Psi}(\zeta) \begin{pmatrix}
1 & \i \\ 0&1
\end{pmatrix}\e^{\i \Theta(\zeta) \sigma_3}, & \text{for } \zeta \in S_8,
\end{array}\right.
\end{equation}
where the regions $S_j$ are the ones depicted in Figure \ref{contourPII3}. It is straightforward to verify that $A(\zeta)$ satisfies the following Riemann-Hilbert Problem
\begin{enumerate}
\item $A(\zeta)$ is analytic for every $\zeta \in \mathbb{C}\setminus  (\mathbb{R} \cup \i \mathbb{R})$ and has continuous boundary values $A_{\pm}$ satisfying the jump relations
$$
A_+(\zeta)=
A_-(\zeta)
\times 
\begin{cases}
\begin{pmatrix}
1 & s \e^{-2 \i\Theta(\zeta)} \\ 0 & 1
\end{pmatrix}, & \zeta\in (0, \i \infty), \\
\begin{pmatrix}
1 & 0 \\ s \e^{2\i\Theta(\zeta)}& 1
\end{pmatrix}, & \zeta\in (-\i \infty, 0), \\
\begin{pmatrix}
0& -\i \e^{-2\i\Theta(\zeta)} \\ -\i \e^{2\i\Theta(\zeta)}& 1
\end{pmatrix}, & \zeta\in (0, \infty), \\
\begin{pmatrix}
1& -\i \e^{-2\i\Theta(\zeta)} \\ -\i \e^{2\i\Theta(\zeta)}& 0
\end{pmatrix}, & \zeta\in (-\infty,0). \\
\end{cases}
$$
\item As $\zeta\to \infty$,
$$
A(\zeta)= O(\zeta^{-1}).
$$
\item $A(\zeta)$ has the following behavior as $\zeta \to 0$
\begin{align*}
A(\zeta)= \left\lbrace \begin{array}{ll}
O\begin{pmatrix}
|\zeta|^{1/2} & |\zeta|^{-1/2} \\
|\zeta|^{1/2} & |\zeta|^{-1/2}
\end{pmatrix},  \hspace{0.3cm}& \im\zeta >0, \\
O\begin{pmatrix}
|\zeta|^{-1/2} & |\zeta|^{1/2} \\
|\zeta|^{-1/2} & |\zeta|^{1/2}
\end{pmatrix} &  \im\zeta <0.
\end{array}\right.
\end{align*}
\end{enumerate}

This new problem enjoys remarkable symmetries. In particular, $$\overline{\e^{-\i\Theta(\zeta)}} = \e^{-\i\Theta(\zeta)} \;  \text{and} \; \e^{-\i\Theta(\bar{\zeta})} = \e^{-\i\Theta(-\zeta)} = \e^{\i\Theta(\zeta)}$$ for every $\zeta \in \i\mathbb{R}$, and $\overline{\e^{-\i\Theta(\zeta)}} = \e^{\i\Theta(\zeta)}$ for every $\zeta \in \mathbb{R}$. For simplicity, we represent the solution to $A(\zeta)$ piecewise by the matrix-functions $A_1, A_2, A_3, A_4$ which are analytic in the sectors $Q_1, Q_2, Q_3, Q_4$ (see Figure \ref{vl:contour2}) respectively, that is,
\begin{align*}
A(\zeta)=
\begin{cases}
A_1(\zeta), & \zeta\in Q_1, \\
A_2(\zeta), & \zeta\in Q_2,\\
A_3(\zeta), & \zeta\in Q_3,\\
A_4(\zeta), & \zeta\in Q_4,
\end{cases}, \qquad \begin{array}{c}
A_1(\zeta) = A_4(\zeta) J_1 \\
A_2(\zeta) = A_1(\zeta) J_2 \\
A_3(\zeta) = A_2(\zeta) J_3 \\
A_4(\zeta) = A_3(\zeta) J_4
\end{array},
\end{align*}
where
\begin{align*}
J_1 = \begin{pmatrix}
0& -\i \e^{-2\i\Theta(\zeta)} \\ -\i \e^{2\i\Theta(\zeta)}& 1
\end{pmatrix}& \qquad J_2 = \begin{pmatrix}
1 & s \e^{-2\i\Theta(\zeta)} \\ 0 & 1
\end{pmatrix} \\
J_3 = \begin{pmatrix}
1& -\i \e^{-2\i\Theta(\zeta)} \\ -\i \e^{2\i\Theta(\zeta)}& 0
\end{pmatrix}& \qquad J_4 = \begin{pmatrix}
1 & 0 \\ s \e^{2\i\Theta(\zeta)}& 1
\end{pmatrix}.
\end{align*}

\begin{figure}[h]
\centering
\begin{tikzpicture}[x=0.75pt,y=0.75pt,yscale=-1,xscale=1]

\draw    (171.5,102) -- (171.47,99.76) ;
\draw [shift={(171.44,97.76)}, rotate = 89.13] [color={rgb, 255:red, 0; green, 0; blue, 0 }  ][line width=0.75]    (10.93,-3.29) .. controls (6.95,-1.4) and (3.31,-0.3) .. (0,0) .. controls (3.31,0.3) and (6.95,1.4) .. (10.93,3.29)   ;
\draw    (171,197.5) -- (170.96,198.32) ;
\draw [shift={(170.87,200.32)}, rotate = 272.62] [color={rgb, 255:red, 0; green, 0; blue, 0 }  ][line width=0.75]    (10.93,-3.29) .. controls (6.95,-1.4) and (3.31,-0.3) .. (0,0) .. controls (3.31,0.3) and (6.95,1.4) .. (10.93,3.29)   ;
\draw  [dash pattern={on 4.5pt off 4.5pt}]  (171.62,43.82) -- (170.69,255.5) ;
\draw    (236.5,149.5) -- (237.52,149.54) ;
\draw [shift={(239.51,149.63)}, rotate = 182.41] [color={rgb, 255:red, 0; green, 0; blue, 0 }  ][line width=0.75]    (10.93,-3.29) .. controls (6.95,-1.4) and (3.31,-0.3) .. (0,0) .. controls (3.31,0.3) and (6.95,1.4) .. (10.93,3.29)   ;
\draw  [dash pattern={on 4.5pt off 4.5pt}]  (288.69,149.93) -- (53.62,149.39) ;
\draw  [dash pattern={on 4.5pt off 4.5pt}]  (100,149.5) -- (96.41,149.58) ;
\draw [shift={(94.41,149.63)}, rotate = 358.69] [color={rgb, 255:red, 0; green, 0; blue, 0 }  ][line width=0.75]    (10.93,-3.29) .. controls (6.95,-1.4) and (3.31,-0.3) .. (0,0) .. controls (3.31,0.3) and (6.95,1.4) .. (10.93,3.29)   ;
\draw  [fill={rgb, 255:red, 0; green, 0; blue, 0 }  ,fill opacity=1 ] (171.14,147.28) .. controls (172.27,147.28) and (173.18,148.34) .. (173.19,149.65) .. controls (173.19,150.97) and (172.28,152.04) .. (171.16,152.04) .. controls (170.04,152.04) and (169.12,150.98) .. (169.12,149.67) .. controls (169.12,148.35) and (170.02,147.28) .. (171.14,147.28) -- cycle ;
\draw   (178.2,55.2) .. controls (203,56.4) and (225,66) .. (242.2,78.4) .. controls (260.2,94) and (277,114.4) .. (279.8,141.2) .. controls (266.6,141.6) and (208.2,141.6) .. (179,141.2) .. controls (178.6,108) and (178.2,78.4) .. (178.2,55.2) -- cycle ;
\draw   (165.04,241.92) .. controls (140.24,240.6) and (111.4,230) .. (98.35,219.6) .. controls (85.8,210.8) and (64.53,182.23) .. (63.86,155.43) .. controls (77.06,155.09) and (135.46,155.38) .. (164.66,155.92) .. controls (164.9,189.12) and (165.15,218.72) .. (165.04,241.92) -- cycle ;
\draw   (279,156) .. controls (277.63,180.79) and (260.2,201.6) .. (248.6,212.8) .. controls (235.8,225.2) and (203.8,240.99) .. (177,241.6) .. controls (176.69,228.4) and (177.59,185.45) .. (178.19,156.25) .. controls (211.39,156.09) and (255.8,155.84) .. (279,156) -- cycle ;
\draw   (63.82,142.2) .. controls (65.07,117.41) and (82.4,96.52) .. (93.95,85.26) .. controls (106.69,72.8) and (138.61,56.86) .. (165.41,56.13) .. controls (165.78,69.33) and (165.09,112.28) .. (164.62,141.48) .. controls (131.42,141.8) and (87.02,142.26) .. (63.82,142.2) -- cycle ;
\draw    (97.5,82.67) -- (95.56,84.08) ;
\draw [shift={(93.95,85.26)}, rotate = 323.85] [color={rgb, 255:red, 0; green, 0; blue, 0 }  ][line width=0.75]    (10.93,-3.29) .. controls (6.95,-1.4) and (3.31,-0.3) .. (0,0) .. controls (3.31,0.3) and (6.95,1.4) .. (10.93,3.29)   ;
\draw    (95,217) -- (96.77,218.37) ;
\draw [shift={(98.35,219.6)}, rotate = 217.82] [color={rgb, 255:red, 0; green, 0; blue, 0 }  ][line width=0.75]    (10.93,-3.29) .. controls (6.95,-1.4) and (3.31,-0.3) .. (0,0) .. controls (3.31,0.3) and (6.95,1.4) .. (10.93,3.29)   ;
\draw    (242.5,218) -- (247.08,214.1) ;
\draw [shift={(248.6,212.8)}, rotate = 139.55] [color={rgb, 255:red, 0; green, 0; blue, 0 }  ][line width=0.75]    (10.93,-3.29) .. controls (6.95,-1.4) and (3.31,-0.3) .. (0,0) .. controls (3.31,0.3) and (6.95,1.4) .. (10.93,3.29)   ;
\draw    (247,82.5) -- (243.72,79.7) ;
\draw [shift={(242.2,78.4)}, rotate = 40.5] [color={rgb, 255:red, 0; green, 0; blue, 0 }  ][line width=0.75]    (10.93,-3.29) .. controls (6.95,-1.4) and (3.31,-0.3) .. (0,0) .. controls (3.31,0.3) and (6.95,1.4) .. (10.93,3.29)   ;

\draw (61.4,52.4) node [anchor=north west][inner sep=0.75pt]    {$Q_{2}$};
\draw (260.63,53.07) node [anchor=north west][inner sep=0.75pt]    {$Q_{1}$};
\draw (61.13,220.63) node [anchor=north west][inner sep=0.75pt]    {$Q_{3}$};
\draw (259.7,221.47) node [anchor=north west][inner sep=0.75pt]    {$Q_{4}$};
\draw (232,93.4) node [anchor=north west][inner sep=0.75pt]    {$\gamma _{1}$};
\draw (90.5,95.4) node [anchor=north west][inner sep=0.75pt]    {$\gamma _{2}$};
\draw (87.5,183.4) node [anchor=north west][inner sep=0.75pt]    {$\gamma _{3}$};
\draw (235,183.9) node [anchor=north west][inner sep=0.75pt]    {$\gamma _{4}$};

\end{tikzpicture}

\caption{Sectors of the complex plane.}
\label{vl:contour2}
\end{figure}
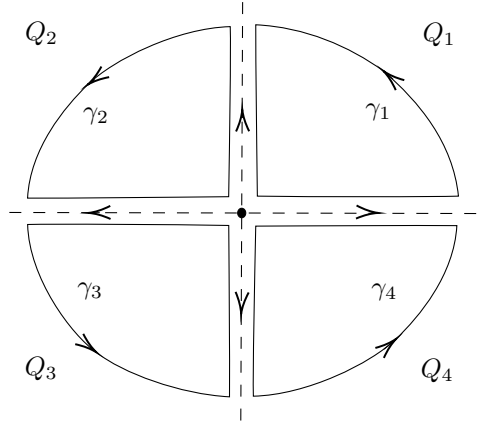

We proceed now as in \cite{CKV08}. Set $B(\zeta):= A(\zeta) \overline{A(\bar{\zeta})}^T$. The integral over $\gamma_1$ (see Figure \ref{vl:contour2}) of $B(\zeta)$ gives
\begin{align}
\int_{\gamma_1} A_1(\zeta) \overline{A_4(\zeta)}^T \dd \zeta= & \lim_{R \to \infty}\left(\int_0^{R} A_1(\zeta) \overline{A_4(\zeta)}^T \dd \zeta+ \int_{\i R}^0 A_1(\zeta) \overline{A_4(\zeta)}^T \dd \zeta\right)=0.
\label{proof:eq1}
\end{align}

On the other side, the integral over $\gamma_4$ of $B(\zeta)$ gives
\begin{align*}
\int_{\gamma_4} B(\zeta) \dd \zeta =& \lim_{R \to \infty}\left(\int_{R}^0 B(\zeta) \dd \zeta + \int_0^{-\i R} B(\zeta) \dd \zeta\right) = \lim_{R \to \infty}\left(\int_{R}^0 B(\zeta) \dd \zeta + \int_{\i R}^0 B(-\zeta) \dd \zeta\right) = 0,
\end{align*}
and, consequently,
\begin{align}
\lim_{R \to \infty}\left(\int_0^{R} A_4(\zeta) \overline{A_1(\zeta)}^T \dd \zeta - \int_{\i R}^0 A_1(\zeta) \overline{A_4(\zeta)}^T \dd \zeta\right)=0.
\label{proof:eq2}
\end{align}

Summing up Equations \eqref{proof:eq1} and \eqref{proof:eq2} and taking the limit, we obtain
\begin{align*}
\int_0^{\infty} A_1(\zeta) \overline{A_4(\zeta)}^T+ A_4(\zeta) \overline{A_1(\zeta)}^T \dd \zeta =& \int_0^{\infty} A_4(\zeta)\left[J_1+\bar{J}_1^T\right] \overline{A_4(\zeta)}^T \dd \zeta \\
=& 2\int_0^{\infty} A_4(\zeta)E_{22} \overline{A_4(\zeta)}^T \dd \zeta =0.
\end{align*}

Analogously, the integrals of $B(\zeta)$ over $\gamma_2$ and $\gamma_3$ give
\begin{align*}
\int_{\gamma_2} B(\zeta) \dd \zeta =&\lim_{R \to \infty}\left(\int_{-R}^0 A_2(\zeta) \overline{A_3(\zeta)}^T \dd \zeta+ \int_0^{\i R} A_2(\zeta) \overline{A_3(\zeta)}^T \dd \zeta\right)=0\\
\int_{\gamma_3} B(\zeta) \dd \zeta =& \lim_{R \to \infty}\left(\int_0^{-R} B(\zeta) \dd \zeta + \int_{-\i R}^0 B(\zeta) \dd \zeta\right) \\
=& \lim_{R \to \infty}\left(\int_0^{-R} B(\zeta) \dd \zeta + \int_0^{\i R} B(-\zeta) \dd \zeta\right)\\
=& \int_0^{-\infty} A_3(\zeta)\overline{A_2(\zeta)}^T \dd \zeta + \int_0^{\i\infty} A_2(\zeta) \overline{A_3(\zeta)}^T \dd \zeta = 0.
\end{align*}

Combining the previous results leads to
\begin{align*}
\int_{-\infty}^0 A_2(\zeta)[J_3+\bar{J}_3^T] \overline{A_2(\zeta)}^T \dd \zeta &= \int_{-\infty}^0 A_4(\zeta)J_1J_2\left[J_3+\bar{J}_3^T\right]\bar{J}_2^T\bar{J}_1^T \overline{A_4(\zeta)}^T \dd \zeta \\
&= 2\int_{-\infty}^0 A_4(\zeta)E_{22} \overline{A_4(\zeta)}^T \dd \zeta = 0.
\end{align*}

Therefore, the entries $(1,1)$ and $(2,2)$ give us
\begin{align*}
\int_{\mathbb{R}} |A_{4,12}(\zeta)|^2+|A_{4,22}(\zeta)|^2\dd \zeta &=0
\end{align*}

This implies that $A_{4,22}(\zeta) = A_{4,12}(\zeta) = 0$. Combined with the jump condition it implies that $A_{1,11}(\zeta) = A_{1,21}(\zeta) = 0$ and that $A_2=A_1$, $A_3=A_4$. It follows that the matrix-valued function $A(\zeta)$ restricted to the upper or lower half-plane is analytic. The proof that the remaining entries are also equal to zero is then exactly the same as in Proposition 2.5 of \cite{CKV08}. Therefore, $A(\zeta) \equiv 0$. \qed

\subsection{Asymptotics for $u$ and $v$}

We are now ready to prove the asymptotics of $u,v$ in the given regime. We resume the results with the following

\begin{proposition}
Let $|\tau_k| \leq M \tau_7^{k/7}$, $k = 1,3,5$ for a large enough constant $M > 0$. Then, as $\tau_7 \to 0$, 
\begin{align*}
u &= a\tau_7^{-1/7}\left(\frac{q}{2}+ p \right)+ O(\tau_7^{1/7})\\
v &= a\tau_7^{-1/7}\left(\frac{\partial}{\partial \tau_1}\frac{q}{2}- a\tau_7^{-1/7}\frac{q^2}{2} + O\left(\tau_7^{1/7}\right)\right),
\end{align*}
where $a=\left(2^6/7\right)^{1/7}$, $q = q(x_0,x_1,x_2)$ is the solution to the third member of the Painlevé II hierarchy characterized by the uniquely solvable Riemann-Hilbert problem \ref{RHPP2} with $s = \i(\iota-1)$, and $p = p(x_0,x_1,x_2)$ satisfies the equation   $\frac{\partial}{\partial x_0} p = \frac{q^2}2$. The variables $x_0,x_1,x_2$ are given by
\begin{equation}
x_0 = - a \frac{\tau_1}{\tau_7^{1/7}}, \quad x_1 = \frac{3a^3}{4}  \frac{\tau_3}{\tau_7^{3/7}}, \quad x_2 = \frac{5a^5}{16} \frac{\tau_5}{\tau_7^{5/7}}.
\end{equation}

Moreover, for $\iota=1$, $q$ is the solution to the third member of the Painlevé II whose asymptotic behavior has been studied in \cite{CIK}, Theorem 1.8.
\end{proposition}
\begin{proof}

As for $u$, we start recalling the equation \eqref{eq:uvsX}. On the other hand, thanks to \eqref{eq:defZ}, 
$X^{(1)}= s_0^{-1/7} Z^{(1)}$, so that
\begin{align*}
u &= s_0^{-1/7}(Z^{(1)}_{11}+\i Z^{(1)}_{12}).
\end{align*}

Now we use the fact that, thanks to the small norm theorem, from $J_R = I + O(s_0^{2/7})$ it follows that $R^{(1)} = O(s_0^{2/7})$, and therefore 
\begin{align*}
Z^{(1)}_{11}+\i Z^{(1)}_{12} =& \tilde{Z}_{11}^{(1)}+\i\tilde{Z}_{12}^{(1)}+ O(s_0^{2/7}).
\end{align*}

The relation given in Proposition \ref{relationZXi} together with the equation \eqref{eq:asympXi} gives 
\begin{align*}
\tilde{Z}^{(1)} = & -2^{5/7} \i \e^{\pi \i \sigma_3/4} \Psi^{(1)}\e^{-\pi \i \sigma_3/4} = -2^{5/7} \i \begin{pmatrix}
\i p & q/2 \\ q/2 & -\i p
\end{pmatrix},
\end{align*}
and the asymptotics for $u$ is proven. The formula for $v$ is proven analogously, starting with \eqref{eq:vvsX} instead of \eqref{eq:uvsX}. 

\end{proof}

\section{The third asymptotic regime}\label{sec:6}

The last asymptotic regime we study is the one where there exist a constant $M_1$ and a (sufficiently large) constant $M_2$  such that $-M_1 s_0^{4/7} \leq \tau_5 \leq -M_2 s_0^{5/7}$. In this Section, for simplicity, we restrict our study to the case $\iota = 1$. It is convenient to re-write the asymptotic expansion \eqref{eq:asympZ} for the Riemann-Hilbert problem \ref{RHPZ} in the following way
\begin{equation}\label{eq:asympZsect7}
\begin{split}Z(\xi) = & \left[I+\frac{\hat{Z}^{(1)}}{\xi}+O(\xi^{-3/2})\right]\xi^{-\frac{\sigma_3}{4}}N\e^{-\left(\frac{2}{7}\xi^{\frac{7}{2}}+\kappa_0\xi^{\frac{5}{2}}+\kappa_1\xi^{\frac{3}{2}}+\kappa_2\xi^{\frac{1}{2}}\right)\sigma_3}\\
&\hspace{5cm}
\times \left\lbrace \begin{array}{cc}
I \hspace{0.3cm}&  |\arg \xi|<\pi-\epsilon\\
\begin{pmatrix}
1 & 0 \\ \pm 1 & 1
\end{pmatrix}  \hspace{0.3cm}& \pi-\epsilon<\pm \arg \xi<\pi  
\end{array}\right.,\end{split}
\end{equation}
where $\kappa_0 = {\tau_5}/{s_0^{5/7}}$, $\kappa_1 = {\tau_3}/{s_0^{3/7}}$ and $\kappa_2 = {\tau_1}/{s_0^{1/7}}$,
as our results are more conveniently expressed in function of the matrix $\hat Z^{(1)}$. 

We define the variable 
\begin{equation}\label{eq:defeta}
	\eta := -\left(\frac{2}{7}\xi^{3}s_0^{1/7}+\frac{\tau_5}{s_0^{4/7}}\xi^{2}+\frac{\tau_3}{s_0^{2/7}}\xi+\tau_1\right),
\end{equation} 
so that the exponential phase in the asymptotic expansion above can be rewritten as $(s_0^{-2/7}\xi)^{1/2}\eta$. The hypothesis in Theorem \ref{thm:asympuvgamma} allows us to control the behavior of $\eta$ when $\xi$ is bounded, under some assumption on the behavior of the variables $\tau_j$.
\begin{assumption}\label{assumptionsthirdregime}
	There exist three constants $M_1,M_2$ and $M_3$, with $M_2$ and $M_3$ sufficiently large, such that $-M_1 s_0^{4/7} \leq \tau_5 \leq -M_2 s_0^{5/7}$, $| \tau_3 | \leq M_1 s_0^{2/7}$, $-M_3 \leq \tau_1 \leq x$, where $$x=(-2s_0^{1/7}/7+\tau_5/s_0^{4/7}-|\tau_3|/s_0^{2/7})$$ and $(-M_3, x) \neq \emptyset$.
\end{assumption}
\begin{lemma}\label{lemma:eta}
Suppose that Assumptions \ref{assumptionsthirdregime} are verified.
Then, for $|\xi| \leq 1$, $\eta$, defined in \eqref{eq:defeta}, has the following properties:
\begin{itemize}
\item[a)] $\eta$ is bounded and has positive real part;
\item[b)] as $s_0 \to 0$, $\eta = -\tau_1 +O(x)$.
\item[c)] $\eta \to 0$ only when $\tau_1 \to x$ and $x \to 0$.
\end{itemize}
\end{lemma}
\begin{proof}
For boundedness and positiveness, notice that under the conditions above
\begin{align*}
-\frac{2}{7}\xi^{3}s_0^{1/7}+\frac{M_2 s_0^{5/7}}{s_0^{4/7}}\xi^{2}-\frac{M_1 s_0^{2/7}}{s_0^{2/7}}\xi -x\leq \eta \leq -\frac{2}{7}\xi^{3}s_0^{1/7}+\frac{M_1 s_0^{4/7}}{s_0^{4/7}}\xi^{2}+\frac{M_1 s_0^{2/7}}{s_0^{2/7}}\xi-M_3,\\
\frac{2}{7}(1-\xi^{3})s_0^{1/7}+M_2 s_0^{1/7}(1+\xi^{2})+M_1(1-\xi) \leq \eta \leq -\frac{2}{7}\xi^{3}s_0^{1/7}+M_1 \xi^{2}+ M_1 \xi+M_3.
\end{align*}
Moreover, denoting $\tilde{\eta} = -\frac{2}{7}\xi^{3}s_0^{1/7}-\frac{\tau_5}{s_0^{4/7}}\xi^{2}-\frac{\tau_3}{s_0^{2/7}}\xi$ we see that $\eta = \tilde{\eta}-\tau_1$, $\tilde{\eta} = O(x)$ and $|\tilde{\eta}|\leq |\tau_1|$. This means that $\eta \to 0$ only on the regime $\tau_1 \to x$ and $x \to 0$. Because of the the term  $\frac{2}{7}\xi^{3}s_0^{1/7}$, we conclude that in such regime $\eta \to 0$ at most as $s_0^{1/7}$. This happens only when $\tau_5 \to -M_2 s_0^{5/7}$, $\tau_3 = O(s_0^{3/7})$ and $\tau_1 \to x$.
\end{proof}

The local parametrix in this asymptotic regime will be given by the following 

\begin{problem}\label{RHPHsigma}\hfill
\begin{enumerate}
\item $H_{\sigma}(\zeta)$ is analytic on $\mathbb{C}\setminus \mathbb R$, and has continuous boundary values $H_{\sigma,\pm}$ satisfying the jump relation
\begin{equation}
H_{\sigma,+}(\zeta)=
H_{\sigma,-}(\zeta)\begin{pmatrix}
1 & 1-\sigma(\zeta) \\ 0 & 1
\end{pmatrix},\quad \zeta\in \mathbb R.
\end{equation}
\item As $\zeta\to \infty$,
\begin{equation}\label{eq:asympHsigma}
\begin{split}H_{\sigma}(\zeta) = & \left[I+\frac{1}{\zeta}\begin{pmatrix}
q_{\sigma} & p_{\sigma} \\ r_{\sigma} & -q_{\sigma}
\end{pmatrix} +O(\zeta^{-2})\right]\zeta^{-\frac{\sigma_3}{4}}N\e^{\zeta^{1/2}\eta\sigma_3}\\
&\hspace{5cm}
\times \left\lbrace \begin{array}{cc}
I \hspace{0.3cm}&  |\arg \zeta|<\pi-\epsilon\\
\begin{pmatrix}
1 & 0 \\ \pm 1 & 1
\end{pmatrix}  \hspace{0.3cm}& \pi-\epsilon<\pm \arg \zeta<\pi  
\end{array}\right..\end{split}
\end{equation}
\end{enumerate}
\end{problem}
%
%
From Lemma \ref{lemma:eta} it is possible to write the following Taylor expansions for $p_{\sigma}(\eta)$ and $q_{\sigma}(\eta)$:
\begin{align}
p_{\sigma}(\eta) &= p_{\sigma}(-\tau_1)+p_{\sigma}'(-\tau_1)(\eta + \tau_1)+O((\eta + \tau_1)^2) = p_{\sigma}(-\tau_1)+O(x), \label{pexpansion}\\
q_{\sigma}(\eta) &= q_{\sigma}(-\tau_1)+q_{\sigma}'(-\tau_1)(\eta + \tau_1)+O((\eta + \tau_1)^2) = q_{\sigma}(-\tau_1)+O(x) \label{qexpansion}. 
\end{align}
They are particularly useful when $x \to 0$ ($\tau_5 =o(s_0^{4/7})$, $\tau_3 = o(s_0^{2/7})$).\\

Some straightforward manipulations allow us to verify that, if $H_\sigma$ satisfies the Riemann-Hilbert problem \ref{RHPHsigma}, then
\begin{equation}
\Phi_{\sigma}(\zeta) := \e^{-\pi \i \sigma_3/4}\begin{pmatrix}
0 & -\i \\ -\i &0
\end{pmatrix} H_{\sigma}(\zeta),
\label{eq:connectCCR}
\end{equation}
satisfies the Riemann-Hilbert problem in Section 7.1 of \cite{CCR}. Below we transcribe some useful results from \cite{CCR} concerning to the small $\eta$ analysis of $H_{\sigma}(\zeta)$.
\begin{lemma}\label{lemma:ccr1}
Set $z:= \eta^2 \zeta$. Then, for $|z|>1$, it follows that as $\eta \to 0$,
\begin{equation}
H_{\sigma}(z \eta^{-2}) = \eta^{\sigma_3/2} S(z) H(z),
\end{equation}
where
\begin{align*}
H(z) &= \begin{pmatrix}
0& \i \\ \i &0
\end{pmatrix}\e^{\pi \i \sigma_3/4} \Phi_{\mathrm{Be}}(z \eta^{-2}) \\
S(z) &= I+\frac{\eta^2}{z}\begin{pmatrix}
\frac{3m_0(\sigma)}{16} & \frac{m_0(\sigma)}{2} \\ - \frac{9m_0(\sigma)}{128} & -\frac{3m_0(\sigma)}{16}
\end{pmatrix} + O(\eta^4/z),
\end{align*}
for $m_0(\sigma) = \int_{\mathbb{R}}(\chi_{[0,\infty)}(s)-\sigma(s)) \dd s$ and $\Phi_{\mathrm{Be}}(\zeta)$ is the same as in \cite{CCR} (Section 7.1):
\begin{equation}
	\Phi_{\mathrm{Be}}(\zeta) := \begin{pmatrix}
								1 & \frac{3 \i}8 \\
								0 & 1
							\end{pmatrix} \begin{pmatrix}
											\sqrt{\pi\zeta}\mathrm{I}_1(\sqrt{\zeta}) & -\i \sqrt{\frac{\zeta}\pi} \mathrm{K}_1(\sqrt{\zeta})\\
											\\
											-\i \sqrt{\pi}\mathrm{I}_0(\sqrt{\zeta}) & \frac{1}{\sqrt{\pi}} \mathrm{K}_0(\sqrt{\zeta})
										\end{pmatrix}, \quad \zeta \in \mathbb{C}\backslash(-\infty,0],
										\label{eq:defPhiBe}
\end{equation}
where $I_0, I_1, K_0$ and $K_1$ are the modified Bessel functions, see for instance \cite{DLMF}.
\end{lemma}
\begin{proof} It follows from \cite{CCR} (Lemma 7.1 and subsection 7.2) together with Equation \eqref{eq:connectCCR}. \end{proof}
\begin{lemma}
As $\eta \to 0$,
\begin{align*}
p_{\sigma}(\eta) &= \frac{\eta}{2}m_0(\sigma)-\frac{1}{8\eta} + O(\eta^3), &
q_{\sigma}(\eta) &= \frac{3m_0(\sigma)}{16}+\frac{9}{128 \eta^2}+ O(\eta^2),\\
r_{\sigma}(\eta) &= -\frac{9 m_0(\sigma)}{128 \eta}+ \frac{39}{512 \eta^3}+ O(\eta).
\end{align*}
\end{lemma}
\begin{proof} It follows from \cite{CCR} (subsection 7.2) together with the asymptotic expansion for $H_{\sigma}$. \end{proof}

\subsection{Local and global parametrices}

We define the global and the local parametrices as follows
\begin{align*}
\begin{array}{ll}
G(\xi) := \xi^{-\sigma_3/4}N \e^{-\theta(\xi)\sigma_3},  & |\xi|>1\\
L(\xi) := s_0^{-\sigma_3/14} H_{\sigma}(s_0^{-2/7}\xi; \eta), \qquad & |\xi|<1,
\end{array}
\end{align*}
where we recall
$$
N = \frac{1}{\sqrt{2}}\begin{pmatrix}
					1 & 1 \\
					-1 & 1
		\end{pmatrix} \e^{-\pi \i \sigma_3/4}.
$$

The related small norm Riemann-Hilbert problem is given by

\begin{align}\label{eq:defR2}
R(\xi) = \left\{ \begin{array}{ll}
(\tau_1 s_0^{-1/7})^{-\sigma_3/2}Z(\xi)G^{-1}(\xi)(\tau_1 s_0^{-1/7})^{\sigma_3/2}, &\quad |\xi|>1\\
(\tau_1 s_0^{-1/7})^{-\sigma_3/2}Z(\xi)L^{-1}(\xi)(\tau_1 s_0^{-1/7})^{\sigma_3/2}, &\quad |\xi|<1.
\end{array}\right.
\end{align}

In the Lemma below $\Gamma := \left( \mathbb R\setminus [-1,1] \right) \cup S^1$ as in the previous Section.
\begin{lemma}\label{RThirdSection}
The matrix $R$ defined in \ref{eq:defR2} 
solves the following Riemann-Hilbert problem
\begin{enumerate}
\item $R$ is analytic on $\mathbb{C}\setminus \Gamma$, and has continuous boundary values $R_{\pm}$ satisfying the jump relation
\begin{align*}
J_R(\xi) = \left\{ \begin{array}{ll}
(\tau_1 s_0^{-1/7})^{-\sigma_3/2}G_-(\xi)\begin{pmatrix}1& \chi_{[0,\infty)}(\xi)-\sigma(s_0^{2/7} \xi)\\ 0 & 1 \end{pmatrix}G_-^{-1}(\xi)(\tau_1 s_0^{-1/7})^{\sigma_3/2}, & \xi \in \mathbb{R}/[-1,1]\\
(\tau_1 s_0^{-1/7})^{-\sigma_3/2}G(\xi)L^{-1}(\xi)(\tau_1 s_0^{-1/7})^{-\sigma_3/2}, &\quad |\xi|=1.
\end{array}\right.
\end{align*}
\item As $\xi \to \infty$ and $|\arg \xi | < \pi$,
\begin{align*}
R(\xi) = I+ R^{(1)}\xi^{-1} + O(\xi^{-2}).
\end{align*}
\end{enumerate}
\end{lemma}

\begin{proof} 
It is clear, by construction, that $R$ is analytic everywhere, except for $\mathbb R$ and $S^1$. As for the interval $[-1,1]$, using the fact that $Z$ and $L$ have the same jumps, one proves that $R$ does not have any discontinuities (using also arguments similar to the ones in the previous section to show that the points of discontinuities of $\sigma$ are removable singularities, and all of them are in the unit disk when $s_0$ is small enough). 
For $\xi \in \mathbb{R}\setminus [-1,1]$, the jump condition comes from the straightforward computation
\begin{align*}
(\tau_1 s_0^{1/7})^{\sigma_3/2}R_+(\xi)(\tau_1 s_0^{-1/7})^{-\sigma_3/2} =& Z_+(\xi) G_{+}^{-1}(\xi)\\
=& Z_-(\xi)J_Z(\xi)J_G^{-1}(\xi) G_{-}^{-1}\\
=& (\tau_1 s_0^{1/7})^{\sigma_3/2} R_-(\xi) (\tau_1 s_0^{-1/7})^{-\sigma_3/2} G_- (\xi)J_Z(\xi)J_G^{-1}(\xi) G_{-}^{-1}(\xi)\\
=&  (\tau_1 s_0^{1/7})^{\sigma_3/2}R_-(\xi)(\tau_1 s_0^{-1/7})^{\sigma_3/2}G_-(\xi) \begin{pmatrix}1& \chi_{[0,\infty)}(\xi)-\sigma(s_0^{2/7} \xi)\\ 0 & 1 \end{pmatrix} G_{-}^{-1}(\xi),
\end{align*}
and analogously for $\xi \in S^1$. 
For the asymptotic condition, as $\xi \to \infty$
\begin{align*}
R(\xi)=& (\tau_1 s_0^{-1/7})^{-\sigma_3/2}Z(\xi)G^{-1}(\xi)(\tau_1 s_0^{-1/7})^{\sigma_3/2}\\
=& (\tau_1 s_0^{-1/7})^{-\sigma_3/2}\left[I+\frac{\hat{Z}^{(1)}}{\xi}+O(\xi^{-3/2})\right](\tau_1 s_0^{-1/7})^{\sigma_3/2} = I + O(\xi^{-1}),
\end{align*}
as claimed.\end{proof}

\begin{lemma}
Suppose that the Assumptions \ref{assumptionsthirdregime} are satisfied. Then, as $s_0 \to 0$,
$$
\norm{J_R-I}_{L^1(\Gamma)} = O(s_0^{2/7}\tau_1^{-2}), \quad \norm{J_R-I}_{L^2(\Gamma)} = O(s_0^{2/7}\tau_1^{-2}), \quad \norm{J_R-I}_{L^\infty(\Gamma)} = O(s_0^{2/7}\tau_1^{-2})
$$
In the equation above, $\norm{J_R-I}_{L^a(\Gamma)}, \; a = 1,2,\infty$, denotes the maximum of the entry-wise corresponding $L^a(\Gamma)$-norm.
\end{lemma}
\begin{proof} 
For $\xi \in \mathbb{R}/[-1,1]$,
\begin{align*}
J_R - I =&(\tau_1 s_0^{-1/7})^{-\sigma_3/2}G_-(\xi)\begin{pmatrix}0& \chi_{[0,\infty)}(\xi)-\sigma(s_0^{2/7} \xi)\\ 0 & 0 \end{pmatrix}G_-^{-1}(\xi)(\tau_1 s_0^{-1/7})^{\sigma_3/2}
\end{align*}
and, using the explicit expression of $G$ one finds that each entry of $J_R - I$ is less or equal to
$$
 \left|\sigma_1(\xi)-\sigma(s_0^{2/7} \xi)\right| \left| \tau_1^{-1}s_0^{2/7}\xi^{-\frac{1}{2}}\right|\left| \exp\left\{
-\frac{4}{7}\xi^{\frac{7}{2}}-2\kappa_0\xi^{\frac{5}{2}}-2\kappa_1\xi^{\frac{3}{2}}-2\kappa_2\xi^{\frac{1}{2}}\right\} \right|,
$$
that is, the jump is exponentially close to the identity. Consequently, Assumption \ref{assumption2} implies the existence of a constant $c>0$ such that
$$ J_R(\xi) - I = O(\e^{-s_0^{-2/7}c}), \quad \xi \in S^1.$$

For the error in the unit circle we must take into account two different regimes. \\
We first consider the case in which $\eta$ is bounded away from zero. In this case, the asymptotic condition for $H_{\sigma}$ holds uniformly. Then,
\begin{align}\label{eq:expRThirdSection}
J_R(\xi) = (\tau_1 s_0^{-1/7})^{-\sigma_3/2}G(\xi) L^{-1}(\xi)(\tau_1 s_0^{-1/7})^{\sigma_3/2} =& \tau_1^{-\sigma_3/2}\left[I-\frac{s_0^{2/7}H_{\sigma}^{(1)}}{\xi}+O \left(\frac{s_0^{2/7}}{\xi}\right)^{3/2}\right]\tau_1^{\sigma_3/2},
\end{align}
leading to
\begin{align*}
J_R(\xi) =& I+ O(s_0^{2/7}\tau_1^{-1}).
\end{align*}

We now consider the case in which $\eta \to 0$. In this case, by Lemma \ref{lemma:eta}, point c), $\eta = -\tau_1+O(x)$ goes to zero at most as $s_0^{1/7}$.

Hence, using Lemma \ref{lemma:ccr1}, for $M_2$ large enough we have $|\eta^2 s_0^{-2/7} \xi| > 1$. Therefore, 
\begin{align*}
L(\xi)=& s_0^{-\sigma_3/14} H_{\sigma}(s_0^{-2/7}\xi; \eta) \nonumber\\
=& s_0^{-\sigma_3/14} \eta^{\sigma_3/2}S(s_0^{-2/7}\xi \eta^2)H(s_0^{-2/7}\xi \eta^2) \nonumber\\
G(\xi) L^{-1}(\xi) =& \eta^{\sigma_3/2} s_0^{-\sigma_3/14} \left[I-\frac{s_0^{2/7}}{\xi\eta^2}H^{(1)}+O\left(s_0^{4/7}\eta^{-4}\right)\right] \left[I-\frac{s_0^{2/7}}{\xi}S^{(1)}+ O\left(s_0^{4/7}\right)\right] \eta^{-\sigma_3/2} s_0^{\sigma_3/14} \nonumber,
\end{align*}
and the jump becomes

\begin{align}
J_R =& (\tau_1 s_0^{-1/7})^{-\sigma_3/2}G(\xi) L^{-1}(\xi)(\tau_1 s_0^{-1/7})^{\sigma_3/2} \nonumber \\
=& (\eta/\tau_1)^{\sigma_3/2} \left[I-\frac{s_0^{2/7}}{\xi\eta^2}H^{(1)}+O\left(s_0^{4/7}\eta^{-4}\right)\right]\left[I-\frac{s_0^{2/7}}{\xi}S^{(1)}+O\left(s_0^{4/7}\right)\right] (\eta/\tau_1)^{-\sigma_3/2} \nonumber\\
=& I-\frac{s_0^{2/7}}{\xi}\left[\eta^{-2}\left(\frac{\eta}{\tau_1}\right)^{\sigma_3/2} H^{(1)}\left(\frac{\eta}{\tau_1}\right)^{-\sigma_3/2} + \left(\frac{\eta}{\tau_1}\right)^{\sigma_3/2} S^{(1)}\left(\frac{\eta}{\tau_1}\right)^{-\sigma_3/2} \right] + O(s_0^{4/7}\eta^{-4}) \nonumber\\
=& I-\frac{s_0^{2/7}}{\xi}\left[\eta^{-2} \begin{pmatrix}
\frac{9}{128}& -\frac{\eta}{8\tau_1} \\ \frac{39\tau_1}{512\eta} & -\frac{9}{128}
\end{pmatrix} + m_0(\sigma)\begin{pmatrix}
\frac{3}{16} & \frac{\eta}{2\tau_1} \\ - \frac{9\tau_1}{128\eta} & -\frac{3}{16}
\end{pmatrix} \right] + O(s_0^{4/7}\eta^{-4}) \label{eq:expRThirdSection2}\\
=& I + O(s_0^{2/7}\eta^{-2})= I + O(s_0^{2/7}\tau_1^{-2}), \nonumber
\end{align}
and the claimed result follows.
\end{proof}
We can now compute the first term correction in the asymptotic expansion of $R$ in Lemma \ref{RThirdSection}. When $\eta$ is far from zero, using equation \eqref{eq:expRThirdSection} of the previous Lemma and the fact that the jumps on the real line are exponentially small, we obtain
\begin{align*}
R^{(1)}(\tau_1, s_0) =& \frac{1}{2\pi i} \int_{S^1} J_1(s)\dd s\\
=& -\frac{1}{2\pi i} s_0^{2/7}\tau_1^{-\sigma_3/2}H_{\sigma}^{(1)}\tau_1^{\sigma_3/2}\int_{S^1} \frac{1}{s}\dd s+ O(s_0^{4/7}\tau_1^{-4})\\
=& -s_0^{2/7}\tau_1^{-\sigma_3/2}H_{\sigma}^{(1)}\tau_1^{\sigma_3/2} + O(s_0^{4/7}\tau_1^{-1}).
\end{align*}
When $\eta \to 0$, we use instead \eqref{eq:expRThirdSection2} and obtain

\begin{align*}
R^{(1)}(\tau_1, s_0)=& s_0^{2/7}\left[\eta^{-2} \begin{pmatrix}
\frac{9}{128}& -\frac{\eta}{8\tau_1} \\ \frac{39\tau_1}{512\eta} & -\frac{9}{128}
\end{pmatrix} + m_0(\sigma)\begin{pmatrix}
\frac{3}{16} & \frac{\eta}{2\tau_1} \\ - \frac{9\tau_1}{128\eta} & -\frac{3}{16}
\end{pmatrix} \right] + O(s_0^{4/7}\eta^{-4}).
\end{align*}

Combined with the asymptotic expression for $H_{\sigma}$ as $\eta \to 0$, we finally obtain that, uniformly 
in the regime given by Assumption \ref{assumptionsthirdregime},
\begin{equation}\label{eq:R1ThirdSection}
	R^{(1)}(\tau_1, s_0) = s_0^{2/7} \begin{pmatrix}
								q_{\sigma}(\eta) & p_{\sigma}(\eta) \tau_1^{-1}\\
								r_{\sigma}(\eta) \tau_1 & -q_{\sigma}(\eta)
								\end{pmatrix}+ O(s_0^{4/7}\eta^{-4}).
\end{equation}

\subsection{Asymptotics for $u$ and $v$}

Using the small-norm theorem for the matrix $R$ defined in \eqref{eq:defR} we obtain
\begin{align*}
\tilde{Z}^{(1)} &= (\tau_1s_0^{-1/7})^{\sigma_3/2} R^{(1)}(\tau_1, s_0)(\tau_1 s_0^{-1/7})^{-\sigma_3/2}\\
&= s_0^{2/7}s_0^{-\sigma_3/14}H_{\sigma}^{(1)}s_0^{\sigma_3/14}+ O(s_0^{3/7}\tau_1^{-3}).
\end{align*}
where $\tilde Z^{(1)}$ appears in \eqref{eq:asympZsect7}.  Combining the latter with the asymptotic expansion \eqref{eq:asympZ} we obtain
\begin{align*}
Z^{(1)} &= \frac{\tilde{Z}_{12}^{(1)}}{2}\begin{pmatrix}
-1 & \i \\ \i & 1
\end{pmatrix} = \frac{s_0^{1/7}p_{\sigma}}{2}\begin{pmatrix}
-1 & \i \\ \i & 1
\end{pmatrix} + O(s_0^{3/7}\tau_1^{-3}),\\
Z^{(2)} &= \frac{1}{2}\begin{pmatrix}\tilde{Z}_{11}^{(1)}+\tilde{Z}_{22}^{(1)} & \i \tilde{Z}_{11}^{(1)}-\i \tilde{Z}_{22}^{(1)}\\-\i (\tilde{Z}_{11}^{(1)}-\i \tilde{Z}_{22}^{(1)}) & \tilde{Z}_{11}^{(1)}+\tilde{Z}_{22}^{(1)}\end{pmatrix}\\
=& \begin{pmatrix} * & \i s_0^{2/7}q_{\sigma}\\ * & * \end{pmatrix}+ O(s_0^{3/7}\tau_1^{-3}),
\end{align*}
where we denoted with $*$ the entries that are not relevant for the following computations.

The asymptotic behavior of the relevant functions is summarised below.

\begin{lemma}\label{lemma6.8}
Let $-M_1 \tau_7^{4/7} \leq \tau_5 \leq -M_2 \tau_7^{5/7}$, $| \tau_3 | \leq M_1 \tau_7^{2/7}$, $-M_3 \leq \tau_1 \leq x$ where 
$$x=-\left(\frac{2}{7}\right)^{6/7}\tau_7^{1/7}+\tau_5\left(\frac{2}{7\tau_7}\right)^{4/7}-|\tau_3|\left(\frac{2}{7\tau_7}\right)^{2/7},$$ 
$M_1, M_2, M_3$ real positive constants such that $(-M_3,x) \neq \emptyset$. Then, as $\tau_7 \to 0$,
\begin{align*}
u &= -p_{\sigma}(-\tau_1) + O(x),\\
v &= -p_{\sigma}^{2}(-\tau_1) + 2 q_{\sigma}(-\tau_1) + O(x),\\
\end{align*}
\end{lemma}
\begin{proof} 
For $u$, using equation \eqref{eq:uvsX} together with $X_{11}^{(1)}+\i X_{12}^{(1)}= s_0^{-1/7}(Z^{(1)}_{11}+\i Z^{(1)}_{12})$ and the expression above for $Z^{(1)}$ above Lemma \ref{lemma6.8}, we obtain
$$
u = -p_{\sigma}(\eta) + O(\tau_7^{2/7}\tau_1^{-3}) = -p_{\sigma}(-\tau_1) + O(x),
$$
where in the last equation we used Lemma \ref{lemma:eta}.
For $v$, we start with \eqref{eq:vvsX} and, similarly, as before, we obtain
\begin{align*}
v =& s_0^{2/7}[-2\i Z^{(1)}_{11}Z^{(1)}_{12}+2(Z^{(1)}_{12})^2-2\i Z^{(2)}_{12}]\\
=& s_0^{2/7}[-s_0^{2/7}p_{\sigma}^2(\eta)-2s_0^{2/7}q_{\sigma}(\eta)+O(\tau_7^{3/7}\tau_1^{-3})]\\
=& -p_{\sigma}^2(\eta)+2q_{\sigma}(\eta)+O(\tau_7^{1/7}\tau_1^{-3}) = -p_{\sigma}^{2}(-\tau_1) + 2 q_{\sigma}(-\tau_1) + O(x),
\end{align*}
again using Lemma \ref{lemma:eta} in the last equality.
\end{proof}

\section*{Acknowledgement}
We are grateful to Thomas Bothner, Tom Claeys and Giulio Ruzza for very fruitful discussions on preliminary versions of the results presented here. M.C. acknowledges the support of the Centre Henri Lebesgue, program ANR-11-LABX0020-0, and the International Research Project PIICQ, funded by CNRS-Math\'ematiques. C.M.S.P. acknowledges the support of the São Paulo Research Foundation (FAPESP), grants $\#$2021/10819-3 and $\#$2023/14157-0. We also want to thank the two anonymous referees, whose comments greatly improved the quality of the exposition.

\appendix

\section{Explicit computations for the potential KdV equations}
\label{App:pKdV}

The aim of this appendix is to give a full derivation of equations \eqref{eq:ApppKdV}.
Expanding $\left(\frac{\partial}{\partial {\tau_1}} X\right)X^{-1}$ at $\zeta = \infty$ and setting to zero the terms in $\zeta^{-1}$ and $\zeta^{-2}$ the following non-trivial relations between the entries of $X^{(j)}$, defined in \eqref{eq:asympX}, are found:
\begin{align*}
X_{11}^{(2)} =& \frac{\left(X_{11}^{(1)}\right)^2 + \left(X_{12}^{(1)}\right)^2}{2} \\
X_{11}^{(4)}=& -\frac{1}{2}\left(\left(X_{11}^{(2)}\right)^{2}+ \left(X_{12}^{(2)}\right)^2\right)  + X_{11}^{(1)} X_{11}^{(3)}+ X_{12}^{(1)} X_{12}^{(3)}.
\end{align*}
These equations will be used in subsequent simplifications, without mentioning them.
Next, expanding $\left(\frac{\partial}{\partial {\tau_3}} X\right)X^{-1}$ we obtain new expressions for $b_1$ and $b_2$:
\begin{align*}
b_1 &= - 2 \i \left(X_{11}^{(1)}\right)^{2} X_{12}^{(1)}+ 4 X_{11}^{(1)} \left(X_{12}^{(1)}\right)^{2} - 2 \i X_{11}^{(1)} X_{12}^{(2)} + 2\i X_{11}^{(2)} X_{12}^{(1)} \\
&+ 2 \i \left(X_{12}^{(1)}\right)^{3} + 4 X_{12}^{(1)} X_{12}^{(2)} - 2 \i X_{12}^{(3)}\\
b_2 &= 2 \i \left(X_{11}^{(1)}\right)^{3} X_{12}^{(1)} - 6 \left(X_{11}^{(1)}\right)^{2} \left(X_{12}^{(1)}\right)^{2} + 2 \i \left(X_{11}^{(1)}\right)^{2} X_{12}^{(2)} - 4 \i X_{11}^{(1)} X_{11}^{(2)} X_{12}^{(1)} \\
&- 6 \i X_{11}^{(1)} \left(X_{12}^{(1)}\right)^{3} - 8X_{11}^{(1)} X_{12}^{(1)} X_{12}^{(2)} + 2 \i X_{11}^{(1)} X_{12}^{(3)} + 4 X_{11}^{(2)} \left(X_{12}^{(1)}\right)^{2} - 2 \i X_{11}^{(2)} X_{12}^{(2)} \\
&+ 2 \i X_{11}^{(3)} X_{12}^{(1)} + 2 \left(X_{12}^{(1)}\right)^{4} - 6 \i \left(X_{12}^{(1)}\right)^{2} X_{12}^{(2)} - 4 X_{12}^{(1)} X_{12}^{(3)} - 2 \left(X_{12}^{(2)}\right)^{2} \\
&+ 2 \i X_{12}^{(4)} - \frac{\partial}{\partial \tau_{3}} X_{11}^{(1)} - \i \frac{\partial}{\partial \tau_{3}} X_{12}^{(1)}
\end{align*}
as well as the identity
\begin{align*}
\frac{\partial}{\partial \tau_{3}} u =& - \i \left(X_{11}^{(1)}\right)^{2} X_{12}^{(2)} - 2 \i X_{11}^{(1)} X_{12}^{(3)} - 2 \i X_{11}^{(3)} X_{12}^{(1)} - \i \left(X_{12}^{(1)}\right)^{2} X_{12}^{(2)} \\
&+ 4 X_{12}^{(1)} X_{12}^{(3)} - 2 \left(X_{12}^{(2)}\right)^{2} - 2 \i X_{12}^{(4)}.
\end{align*}

The formulas for $b_0$ and $b_2$ allow us to verify the first equation in \eqref{eq:ApppKdV}.
\begin{equation*}
b_2 = \frac{b_0^2}{4} -2\frac{\partial}{\partial \tau_{3}}u.
\end{equation*}

Next, expanding $\left(\frac{\partial}{\partial {\tau_5}} X\right)X^{-1}$, and using relations coming from the fact that the term in $\zeta^{-1}$ is equal to $0$, we further obtain
\begin{align*}
b_3&= \frac{1}{4} \i \left(X_{11}^{(1)}\right)^{4} X_{12}^{(1)} + \frac{1}{2} \i \left(X_{11}^{(1)}\right)^{2} \left(X_{12}^{(1)}\right)^{3} + 2 \left(X_{11}^{(1)}\right)^{2} X_{12}^{(1)} X_{12}^{(2)} - \i \left(X_{11}^{(1)}\right)^{2} X_{12}^{(3)} \\
&- 2 \i X_{11}^{(1)} X_{11}^{(3)} X_{12}^{(1)} + 4 X_{11}^{(1)} X_{12}^{(1)} X_{12}^{(3)} - 2 \i X_{11}^{(1)} X_{12}^{(4)} + 4 X_{11}^{(3)} \left(X_{12}^{(1)}\right)^{2} - 2 \i X_{11}^{(3)} X_{12}^{(2)} \\
&+ \frac{1}{4} \i \left(X_{12}^{(1)}\right)^{5} + 2 \left(X_{12}^{(1)}\right)^{3} X_{12}^{(2)} + 5 \i \left(X_{12}^{(1)}\right)^{2} X_{12}^{(3)} - 3 \i X_{12}^{(1)} \left(X_{12}^{(2)}\right)^{2} + 4 X_{12}^{(1)} X_{12}^{(4)} \\
&- 2 \i X_{12}^{(5)}\\
b_4&= - 2 \i \left(X_{11}^{(1)}\right)^{3} X_{12}^{(1)} + 2 \left(X_{11}^{(1)}\right)^{2} \left(X_{12}^{(1)}\right)^{2} - 2 \i \left(X_{11}^{(1)}\right)^{2} X_{12}^{(2)} + 4 \i X_{11}^{(1)} X_{11}^{(2)} X_{12}^{(1)} \\
&- 2 \i X_{11}^{(1)} \left(X_{12}^{(1)}\right)^{3} - 2 \i X_{11}^{(1)} X_{12}^{(3)} - 4 X_{11}^{(2)} \left(X_{12}^{(1)}\right)^{2} + 2 \i X_{11}^{(2)} X_{12}^{(2)} - 2 \i X_{11}^{(3)} X_{12}^{(1)} \\
&+ 2 \left(X_{12}^{(1)}\right)^{4} - 2 \i \left(X_{12}^{(1)}\right)^{2} X_{12}^{(2)} + 4 X_{12}^{(1)} X_{12}^{(3)} - 2 \left(X_{12}^{(2)}\right)^{2} - 2 \i X_{12}^{(4)}\\
b_6 &=  \left(X_{11}^{(1)}\right)^{4} \left(X_{12}^{(1)}\right)^{2} + 2 \left(X_{11}^{(1)}\right)^{2} \left(X_{12}^{(1)}\right)^{4} - 4 \i \left(X_{11}^{(1)}\right)^{2} \left(X_{12}^{(1)}\right)^{2} X_{12}^{(2)} - 4 \left(X_{11}^{(1)}\right)^{2} X_{12}^{(1)} X_{12}^{(3)} \\
&- 8 X_{11}^{(1)} X_{11}^{(3)} \left(X_{12}^{(1)}\right)^{2} - 8 \i X_{11}^{(1)} \left(X_{12}^{(1)}\right)^{2} X_{12}^{(3)} - 8 X_{11}^{(1)} X_{12}^{(1)} X_{12}^{(4)} - 8 \i X_{11}^{(3)} \left(X_{12}^{(1)}\right)^{3}\\
& - 8 X_{11}^{(3)} X_{12}^{(1)} X_{12}^{(2)} + \left(X_{12}^{(1)}\right)^{6} - 4 \i \left(X_{12}^{(1)}\right)^{4} X_{12}^{(2)} + 4 \left(X_{12}^{(1)}\right)^{3} X_{12}^{(3)} - 4 \left(X_{12}^{(1)}\right)^{2} \left(X_{12}^{(2)}\right)^{2}\\
& - 8 \i \left(X_{12}^{(1)}\right)^{2} X_{12}^{(4)} - 8 X_{12}^{(2)} X_{12}^{(4)} + 4 \left(X_{12}^{(3)}\right)^{2} - 2 \frac{\partial}{\partial \tau_{5}} u.
\end{align*}

Once again, the explicit form for such entries allows us to verify the second equation in \eqref{eq:ApppKdV}:
\begin{equation}
b_6 =-b_1^2+\frac{b_0^3}{8}-b_0b_4 -2\frac{\partial}{\partial \tau_{5}}u. 
\end{equation}
 
Analogously, using the expansion of $\left(\frac{\partial}{\partial {\tau_7}} X\right)X^{-1}$ one gets
\begin{align*}
b_{10}
= & -\frac{1}{2} \left(X_{11}^{(1)}\right)^{6} \left(X_{12}^{(1)}\right)^{2} - \frac{3}{2} \left(X_{11}^{(1)}\right)^{4} \left(X_{12}^{(1)}\right)^{4} + \i \left(X_{11}^{(1)}\right)^{4} \left(X_{12}^{(1)}\right)^{2} X_{12}^{(2)}\\
& + \left(X_{11}^{(1)}\right)^{4} X_{12}^{(1)} X_{12}^{(3)}+ 4 \left(X_{11}^{(1)}\right)^{3} X_{11}^{(3)} \left(X_{12}^{(1)}\right)^{2} - \frac{3}{2}\left(X_{11}^{(1)}\right)^{2} \left(X_{12}^{(1)}\right)^{6}\\
& + 2 \i \left(X_{11}^{(1)}\right)^{2} \left(X_{12}^{(1)}\right)^{4} X_{12}^{(2)} + 6 \left(X_{11}^{(1)}\right)^{2} \left(X_{12}^{(1)}\right)^{3} X_{12}^{(3)} - 2 \left(X_{11}^{(1)}\right)^{2} \left(X_{12}^{(1)}\right)^{2} \left(X_{12}^{(2)}\right)^{2}\\
&- 4\i \left(X_{11}^{(1)}\right)^{2} \left(X_{12}^{(1)}\right)^{2} X_{12}^{(4)} - 4 \left(X_{11}^{(1)}\right)^{2} X_{12}^{(1)} X_{12}^{(5)} + 4 X_{11}^{(1)} X_{11}^{(3)} \left(X_{12}^{(1)}\right)^{4}\\
& - 8 \i X_{11}^{(1)} X_{11}^{(3)} \left(X_{12}^{(1)}\right)^{2} X_{12}^{(2)} - 8 X_{11}^{(1)} X_{11}^{(3)} X_{12}^{(1)} X_{12}^{(3)} - 8 X_{11}^{(1)} X_{11}^{(5)} \left(X_{12}^{(1)}\right)^{2}\\
& - 8 \i X_{11}^{(1)} \left(X_{12}^{(1)}\right)^{2} X_{12}^{(5)} - 8 X_{11}^{(1)} X_{12}^{(1)} X_{12}^{(6)} - 4 \left(X_{11}^{(3)}\right)^{2} \left(X_{12}^{(1)}\right)^{2} \\
&- 8 \i X_{11}^{(3)} \left(X_{12}^{(1)}\right)^{2} X_{12}^{(3)} - 8 X_{11}^{(3)} X_{12}^{(1)} X_{12}^{(4)} - 8 \i X_{11}^{(5)} \left(X_{12}^{(1)}\right)^{3} - 8 X_{11}^{(5)} X_{12}^{(1)} X_{12}^{(2)}\\
& - \frac{1}{2} \left(X_{12}^{(1)}\right)^{8} + \i \left(X_{12}^{(1)}\right)^{6} X_{12}^{(2)} + 5\left(X_{12}^{(1)}\right)^{5} X_{12}^{(3)} - 2 \left(X_{12}^{(1)}\right)^{4} \left(X_{12}^{(2)}\right)^{2}\\
& - 4 \i \left(X_{12}^{(1)}\right)^{4} X_{12}^{(4)} - 8 \i \left(X_{12}^{(1)}\right)^{3} X_{12}^{(2)} X_{12}^{(3)} + 4 \left(X_{12}^{(1)}\right)^{3} X_{12}^{(5)} + 4 \i \left(X_{12}^{(1)}\right)^{2} \left(X_{12}^{(2)}\right)^{3}\\
& - 8 \left(X_{12}^{(1)}\right)^{2} X_{12}^{(2)} X_{12}^{(4)} - 4 \left(X_{12}^{(1)}\right)^{2} \left(X_{12}^{(3)}\right)^{2} - 8 \i \left(X_{12}^{(1)}\right)^{2} X_{12}^{(6)}\\
& + 4 X_{12}^{(1)} \left(X_{12}^{(2)}\right)^{2} X_{12}^{(3)} - 8 X_{12}^{(2)} X_{12}^{(6)} + 8 X_{12}^{(3)} X_{12}^{(5)} - 4 \left(X_{12}^{(4)}\right)^{2}- 2\frac{\partial}{\partial \tau_{7}}u,
\end{align*}
and
\begin{align*}
b_8 =& \frac{1}{4} \i \left(X_{11}^{(1)}\right)^{4} X_{12}^{(2)} + \frac{1}{2} \i \left(X_{11}^{(1)}\right)^{2} \left(X_{12}^{(1)}\right)^{2} X_{12}^{(2)} - \i \left(X_{11}^{(1)}\right)^{2} X_{12}^{(4)} - 2 \i X_{11}^{(1)} X_{11}^{(3)} X_{12}^{(2)}\\
& - 2 \i X_{11}^{(1)} X_{12}^{(5)} - 2 \i X_{11}^{(3)} X_{12}^{(3)} - 2 \i X_{11}^{(5)} X_{12}^{(1)} + \frac{1}{4} \i \left(X_{12}^{(1)}\right)^{4} X_{12}^{(2)} - \i \left(X_{12}^{(1)}\right)^{2} X_{12}^{(4)}\\
& - 2 \i X_{12}^{(1)} X_{12}^{(2)} X_{12}^{(3)} + 4 X_{12}^{(1)} X_{12}^{(5)} + \i \left(X_{12}^{(2)}\right)^{3} - 4 X_{12}^{(2)} X_{12}^{(4)} + 2 \left(X_{12}^{(3)}\right)^{2} - 2 \i X_{12}^{(6)}.
\end{align*}

Substituting $\frac{\partial}{\partial \tau_{5}}u = b_8$ in the formula for $b_6$, one finds the third equation of \eqref{eq:ApppKdV}:
\begin{equation}
b_{10} = \frac{b_0b_6}{2}-2b_1b_3-b_4b_5 -2\frac{\partial}{\partial \tau_{7}}u.
\end{equation}

\section{The Lax pair associated to the $\mathrm{PI}^{(2)}$ equation}
\label{AppendixPI}

The aim of this Appendix, added for the reader's convenience, is to write explicitly the Lax system associated to the Riemann-Hilbert problem \ref{RHP1}, and explain how the equations in Proposition \ref{prop:eqPI2} are obtained. We also gather information about the asymptotics of the density of the point process $\cal X$ on the negative real axis, see Corollary \ref{cor:density} below. Since the jump on the Riemann-Hilbert problem \ref{RHP1} is constant, using the asymptotic expansion \eqref{eq:asympPhi}, one proves that there exist three polynomial (in $z$) matrices $U,W$ and $V$ such that the following system of equations is satisfied:
\begin{equation}
	\begin{cases}
		\dfrac{\partial}{\partial t_0} \Phi(t_0,t_1;z) = W(t_0,t_1;z) \Phi(t_0,t_1;z), \\
		\dfrac{\partial}{\partial t_1} \Phi(t_0,t_1;z) = V(t_0,t_1;z) \Phi(t_0,t_1;z), \\
		\dfrac{\partial}{\partial z} \Phi(t_0,t_1;z) = U(t_0,t_1;z) \Phi(t_0,t_1;z).
	\end{cases}
\end{equation}
These Lax matrices, which can be written in terms of  the $h$ and $y$ appearing in the asymptotic expansion \eqref{eq:asympPhi}, read explicitly
\begin{equation}
	W(z) = \begin{pmatrix}0 & -2\\ \\ - 2 z + 2 y + h_{t_0} & 0\end{pmatrix}, \quad V(z) = \begin{pmatrix}\frac{1}{6}y_{t_0} & \frac{2}{3} z + \frac{2}{3} y \\ \\ \frac{2}{3} z^{2}- \frac{2}{3} z y + \frac{2}{3} y^{2} +2 h_{t_1} & -\frac{1}{6}y_{t_0} \end{pmatrix}
\end{equation}
and
\begin{equation}
	U(z) = \begin{pmatrix}
 \frac{1}{4} y_{t_0} z+\frac{3}4y y_{t_0}+\frac{1}{64}y_{t_0t_0t_0} & z^2+yz+ \frac{3y^2}{2} +\frac{1}{16}y_{t_0t_0} +t_{1} \\
 \\
z^3-yz^2+( t_{1} - \frac{1}{2} y^2 -\frac{1}{16}y_{t_0t_0})z+ 2y^3+\frac{1}{8}y y_{t_0t_0}-\frac{1}{16}y_{t_0}^2 - 2t_{0} & -\frac{1}{4} y_{t_0} z - \frac{3}4y y_{t_0} - \frac{1}{64}y_{t_0t_0t_0}
\end{pmatrix}.
\end{equation}
Moreover, using the fact that $\lim_{z \to \infty} z\left[\left(\frac{\partial}{\partial t_0} \Phi(z)\right)\Phi^{-1}(z) - W(z)\right]_{12} = 0$, we also obtain the relation \eqref{eq:handy} $h_{t_0} = 2y$. Equations \eqref{eq:PI2} and \eqref{eq:KdV} are derived, respectively, from the Lax compatibility conditions 
$$
 \frac{\partial}{\partial z}W - \frac{\partial}{\partial t_0} U + [W,U] = 0 \quad \text{and} \quad \frac{\partial}{\partial t_1} W - \frac{\partial}{\partial t_0} V + [W ,V] = 0. 
$$
As for the potential KdV equation \eqref{eq:pKdV}, using the fact that $$\lim_{z \to \infty} z\left[\dfrac{\partial}{\partial z} \Phi(t_0,t_1;z)\Phi^{-1}(t_0,t_1,z) - U(z) \right]  = 0,$$ 
by inspection of the entries we find that $3 y^2 + \frac{1}8 y_{t_0} + 3h_{t_1} = 0$, which gives \eqref{eq:pKdV} using equation \eqref{eq:handy}.

The explicit form of the Lax matrices associated to $\Phi$ gives also informations about the density of the point process with kernel \eqref{eq:introkernel}.
\begin{corollary}\label{cor:density}
	$K(u,u) = \frac{|u|^{\frac{5}{2}}}{\pi}+O(|u|^{1/2})$ as $u \to -\infty$.
\end{corollary}
\begin{proof}
We have
\begin{align*}
K(u,u) &= \lim_{v \to u} \frac{\phi_1(u)\phi_2(v)-\phi_1(v)\phi_2(u)}{-2\pi i (u-v)}\\
&= \frac{1}{2\pi i}\lim_{v \to u} \phi_1(u)\frac{\partial}{\partial v}\phi_2(v)-\frac{\partial}{\partial v}\phi_1(v)\phi_2(u)\\
&= \frac{1}{2\pi i}\left\{\left(u^3-yu^2+\left(t_{1} - \frac{1}{2} y^2 -\frac{1}{16}y_{t_0t_0}\right)u+ 2y^3+\frac{1}{8}y y_{t_0t_0}-\frac{1}{16}y_{t_0}^{2} - 2t_{0} \right)\phi_1^2(u) \right. + \\
& \left. -2\left ( \left(  \frac{1}{4} y_{t_0}\right)u - y_{t_1}\right)\phi_1(u)\phi_2(u)- \left(u^2+yu+\frac{3y^2}{2} +\frac{1}{16} y_{t_0t_0} +t_{1}\right)\phi_2^2(u)\right\}.
\end{align*}
At last, from Equation \eqref{eq:asympPhi}, one obtains
\begin{align*}
2\pi \i K(u,u) =&  2 \i |u|^{\frac{5}{2}}+ 2 \i|u|^{\frac{1}{2}} t_{1} + O(|u|^{-\frac{1}{2}}).
\end{align*}
\end{proof}

\end{document}